%% file: main.tex
\documentclass[runningheads]{llncs} 
\usepackage[T1]{fontenc} 
\usepackage{amsmath}
\usepackage{amssymb}
\usepackage{marginnote}
\usepackage{thm-restate}
\usepackage{mathtools}
\usepackage{complexity}
\usepackage{xparse}
\usepackage{hyperref}

\usepackage{color}

\let\doendproof\endproof
\renewcommand\endproof{~\hfill$\qed$\doendproof}

\usepackage[maxbibnames=99,style=numeric,citestyle=numeric-comp,backend=biber]{biblatex}
\usepackage{bm} 
\usepackage{todonotes} 
\usepackage{booktabs} 
\usepackage{enumitem} 
\usepackage{totcount}
\usepackage{newfile}
\usepackage[capitalize]{cleveref}

\regtotcounter{@todonotes@numberoftodonotes}

\input{macros}

\title{Unboundedness problems for machines with reversal-bounded counters}

\author{Pascal Baumann\inst{1}\orcidID{0000-0002-9371-0807}
\and Flavio D'Alessandro\inst{2}
\and Moses Ganardi\inst{1}\orcidID{0000-0002-0775-7781}
\and Oscar Ibarra\inst{3}
\and Ian McQuillan\inst{4}
\and Lia Sch\"{u}tze\inst{1}\orcidID{0000-0003-4002-5491}
\and Georg Zetzsche\inst{1}\orcidID{0000-0002-6421-4388}}

\authorrunning{P. Baumann et al.}

\institute{Max Planck Institute for Software Systems (MPI-SWS), Germany
\and Dept. of Mathematics G. Castelnuovo, Sapienza University of Rome, Italy 
\and Dept. of Computer Science, University of California, Santa Barbara, CA, USA 
\and Dept. of Computer Science, University of Saskatchewan, Saskatoon, Canada}

\begin{document}

\maketitle
\begin{abstract}
	We consider a general class of decision problems concerning formal
	languages, called ``(one-dimensional) unboundedness predicates'', for automata
	that feature reversal-bounded counters (RBCA).  We show that each problem in this
	class reduces---non-deterministically in polynomial time---to the same problem
	for just finite automata. We also show an analogous reduction for automata that
	have access to both a pushdown stack and reversal-bounded counters (PRBCA).

	This allows us to answer several open questions: For example, we show
	that it is $\coNP$-complete to decide whether a given (P)RBCA language
	$L$ is bounded, meaning whether there exist words $w_1,\ldots,w_n$ with
	$L\subseteq w_1^*\cdots w_n^*$.  For PRBCA, even decidability was open. Our
	methods also show that there is no language of a (P)RBCA of intermediate
	growth. This means, the number of words of each length grows either polynomially or exponentially.
	Part of our proof is likely of independent interest: We show that one
	can translate an RBCA into a machine with $\Z$-counters in logarithmic space,
	while preserving the accepted language.

\end{abstract}
\section{Introduction}
\input{introduction.tex}
\section{Main Results: Unboundedness and (P)RBCA}\label{sec:results}
\input{results.tex}
\section{Translating reversal-bounded counters into $\Z$-counters}
\label{sec:rbc-zvass-conversion}
\input{rbc-zvass-conversion.tex}
\section{Deciding unboundedness predicates}
\label{sec:unboundedness-reduction}
\input{unboundedness-reduction.tex}

\section{Growth}
\label{sec:growth}
\input{growth}

\subsubsection{Acknowledgments} We are grateful to Manfred Kufleitner for sharing
the manuscript~\cite{KufleitnerEuler} before it was publicly available.  It
provides an alternative proof for constructing an existential Presburger
formula for the Parikh image of a context-free grammar.  The latter was also
shown in \cite{esparza1997petri,DBLP:conf/cade/VermaSS05}. We use it in
\cref{euler}, which could also be derived from \cite[Theorem
3.1]{esparza1997petri}. However, we provide a simple direct proof of
\cref{euler} inspired by Kufleitner's proof.

\marginnote{\includegraphics[width=3cm]{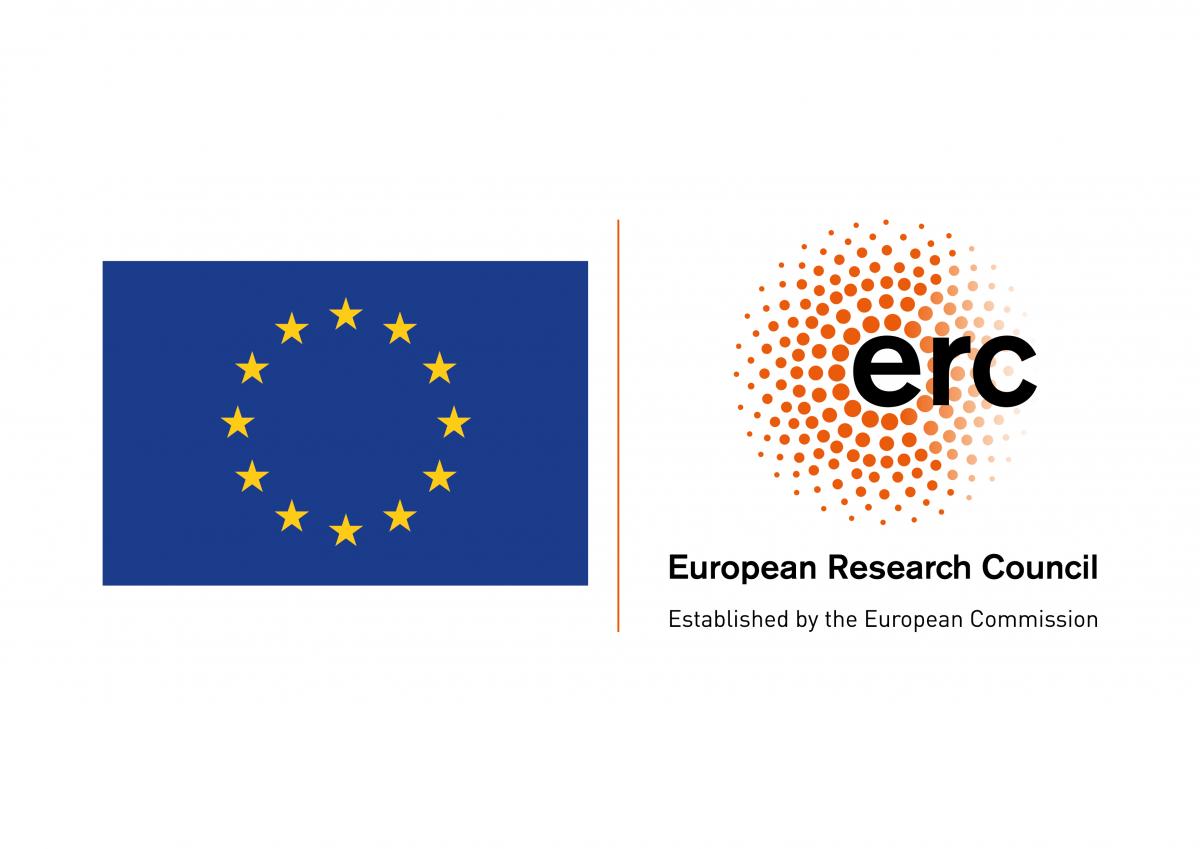}}This work is funded by the European Union (ERC, FINABIS, 101077902). Views and opinions expressed are however those of the author(s) only and do not necessarily reflect those of the European Union or the European Research Council Executive Agency. Neither the European Union nor the granting authority can be held responsible for them.

\label{beforebibliography}
\newoutputstream{pages}
\openoutputfile{main.pages.ctr}{pages}
\addtostream{pages}{\getpagerefnumber{beforebibliography}}
\closeoutputstream{pages}
\printbibliography

\pagebreak

\appendix
\section{Additional material on \cref{sec:results}}
\label{appendix-results}
\input{appendix-results}

\section{Additional material on \cref{sec:rbc-zvass-conversion}}
\label{app:1-reversal-construction}
\input{app-rbc-reversal-reduction.tex}

\section{Additional material on \cref{sec:unboundedness-reduction}}
\label{appendix-unboundedness-reduction}
\input{appendix-unboundedness-reduction}

\section{Additional material on \cref{sec:growth}}
\label{appendix-growth}

\input{appendix-growth}

\newoutputstream{todos}
\openoutputfile{main.todos.ctr}{todos}
\addtostream{todos}{\arabic{@todonotes@numberoftodonotes}}
\closeoutputstream{todos}

\end{document}

%% file: macros.tex
\newcommand{\Z}{\mathbb{Z}}
\newcommand{\N}{\mathbb{N}}

\newcommand{\ZVASS}{\text{$\mathbb{Z}$-$\mathsf{VASS}$}}
\newcommand{\PRBCA}{\mathsf{PRBCA}}
\newcommand{\RBCA}{\mathsf{RBCA}}
\newcommand{\PDA}{\mathsf{PDA}}
\newcommand{\NFA}{\mathsf{NFA}}

\newcommand{\cA}{\mathcal{A}}
\newcommand{\cV}{\mathcal{V}}

\newcommand{\cC}{\mathcal{C}}
\newcommand{\cM}{\mathcal{M}}
\newcommand{\cT}{\mathcal{T}}
\newcommand{\bigO}{\mathcal{O}}
\renewcommand{\vec}[1]{\bm{#1}}
\newcommand{\var}[1]{\mathtt{#1}}

\makeatletter
\newcommand{\oset}[3][0ex]{%
  \mathrel{\mathop{#3}\limits^{
    \vbox to#1{\kern-2\ex@
    \hbox{$\scriptstyle#2$}\vss}}}}
\makeatother

\DeclareDocumentCommand{\autstep}{O{}}{\oset{#1}{\rightarrow}}
\DeclareDocumentCommand{\autsteps}{O{}}{\oset{#1}{\rightarrow}}

\DeclareDocumentCommand{\deriv}{O{}}{\,\Rightarrow_{#1}\,}
\DeclareDocumentCommand{\derivs}{O{}}{\,\oset{*}{\Rightarrow}_{#1}\,}
\DeclareDocumentCommand{\derivv}{O{} m}{\,\Rightarrow_{#1}^{#2}\,}
\DeclareDocumentCommand{\derivvs}{O{} m}{\,\oset{*}{\Rightarrow}_{#1}^{#2}\,}

\newcommand{\yield}{\mathsf{yield}}

\newcommand{\App}{\mathrm{App}}

\newcommand{\pP}{\mathfrak{p}}
\newcommand{\bu}{\bm{u}}
\newcommand{\bv}{\bm{v}}
\newcommand{\bmf}{\bm{f}}
\newcommand{\bzero}{\bm{0}}

\newcommand{\Parikh}{\Psi}
\newcommand{\counters}{\varphi}
\newcommand{\pumpleq}{\sqsubseteq}
\newcommand{\linembed}{\hookrightarrow_{\mathrm{lin}}}

\newcommand{\ctr}{\mathsf{c}}
\newcommand{\pctr}{\mathsf{p}}
\newcommand{\upcl}[1]{#1\mathop{\uparrow}}

\newcommand{\bin}{\mathsf{bin}}

%% file: introduction.tex
A classic idea in the theory of formal languages is the concept of
boundedness of a language.  A language $L$ over an alphabet $\Sigma$ is called
\emph{bounded} if there exists a number $n\in\N$ and words
$w_1,\ldots,w_n\in\Sigma^*$ such that $L\subseteq w_1^*\cdots w_n^*$. What
makes boundedness important is that a rich variety of algorithmic problems
become decidable for bounded languages.  For example, when Ginsburg and
Spanier~\cite{ginsburg1964bounded} introduced boundedness in 1964, they already
showed that given two context-free languages, one of them bounded, one can decide
inclusion~\cite[Theorem 6.3]{ginsburg1964bounded}. This is because if
$L\subseteq w_1^*\cdots w_n^*$ for a context-free language, then the set
$\{(x_1,\ldots,x_n)\in\N^n \mid w_1^{x_1}\cdots w_n^{x_n}\in L\}$ is
effectively semilinear, which permits expressing inclusion in Presburger
arithmetic. Here, boundedness is a crucial assumption: Hopcroft
has shown that if $L_0\subseteq\Sigma^*$ is context-free, then the problem
of deciding $L_0\subseteq L$ for a given context-free language $L$ is decidable if
and only if $L_0$ is bounded~\cite[Theorem 3.3]{hopcroft1969equivalence}.

The idea of translating questions about bounded languages into
Presburger arithmetic has been applied in several other contexts. For example,
Esparza, Ganty, and Majumdar~\cite{DBLP:conf/lics/EsparzaGM12} have shown that
many classes of infinite-state systems are \emph{perfect modulo bounded
languages}, meaning that the bounded
languages form a subclass that is amenable to many algorithmic problems. As
another example, the subword ordering has a decidable first-order theory on
bounded context-free languages~\cite{DBLP:conf/fossacs/KuskeZ19}, whereas on
languages $\Sigma^*$, even the existential theory is
undecidable~\cite{DBLP:conf/lics/HalfonSZ17}.  This, in turn, implies that
initial limit Datalog is decidable for the subword ordering on bounded
context-free languages~\cite{DBLP:conf/lics/BurnORW21}. Finally, bounded context-free languages
can be closely approximated by regular ones~\cite{DBLP:journals/tcs/DAlessandroI15b}.


This raises the question of how one can decide whether a given language is
bounded.  For context-free languages this problem is
decidable~\cite[Theorem~5.2(a)]{ginsburg1964bounded} in polynomial
time~\cite[Theorem 19]{DBLP:journals/ijfcs/GawrychowskiKRS10}.

\vspace{-0.2cm}
\subsubsection{Boundedness for RBCA.} 
Despite the importance of boundedness, it had been open for many
years~\cite{DBLP:journals/ijfcs/CadilhacFM12,DBLP:conf/dlt/EremondiIM15}\footnote{Note
	that \cite{DBLP:journals/ijfcs/CadilhacFM12} is about Parikh automata,
which are equivalent to RBCA.} whether boundedness is
decidable for one of the most well-studied types of infinite-state systems:
\emph{reversal-bounded (multi-)counter automata} (RBCA). These are machines
with counters that can be incremented, decremented, and even tested for zero.
However, in order to achieve decidability of basic questions, there is a bound
on the number of times each counter can \emph{reverse}, that is, switch between
\emph{incrementing} and \emph{decrementing phases}. They were first studied in
the 1970s~\cite{DBLP:journals/jacm/Ibarra78,baker1974reversal} and have
received a lot of attention since~\cite{GURARI1981220,DBLP:journals/jcss/Chan88,DBLP:journals/tcs/IbarraSDBK02,DBLP:conf/mfcs/FinkelS08,DBLP:conf/cav/HagueL11,DBLP:journals/ita/CadilhacFM12,DBLP:journals/ijfcs/CadilhacFM12,DBLP:journals/ijfcs/EremondiIM18,DBLP:journals/iandc/JantzenK03,DBLP:conf/icalp/Zetzsche16,DBLP:conf/icalp/ClementeCLP17,DBLP:journals/mst/CadilhacKM18,carpi2021relationships,DBLP:conf/lics/HalfonSZ17,KlaedtkeRuess2003}. The desirable properties mentioned above for bounded context-free
languages also apply to bounded RBCA.
Furthermore, any bounded language accepted by an RBCA (even one augmented with a stack) can be effectively determinized~\cite{DBLP:journals/ijfcs/IbarraS12} (see also~\cite{DBLP:journals/ijfcs/CadilhacFM12,carpi2021relationships}),
opening up even more avenues to algorithmic analysis. This
makes it surprising that decidability of boundedness remained open for many years.

Decidability of boundedness for RBCA was settled
in~\cite{DBLP:conf/icalp/CzerwinskiHZ18}, which proves boundedness decidable
even for the larger class of vector addition systems with states (VASS), with
acceptance by configuration. However, the results from
\cite{DBLP:conf/icalp/CzerwinskiHZ18} leave several aspects unclarified, which we
investigate here:
\begin{enumerate}[leftmargin=*, label=Q\arabic*:]
	\item What is the complexity of deciding boundedness for RBCA? The
		algorithm in~\cite{DBLP:conf/icalp/CzerwinskiHZ18} employs the
		KLMST decomposition for
		VASS~\cite{DBLP:conf/stoc/Mayr81,DBLP:conf/stoc/Kosaraju82,DBLP:journals/tcs/Lambert92,sacerdote1977decidability,leroux2015reachability},
		which is well-known to incur Ackermannian complexity~\cite{DBLP:conf/lics/LerouxS19}.
	\item Is boundedness decidable for \emph{pushdown RBCA}
		(PRBCA)~\cite{DBLP:journals/jacm/Ibarra78}?
		These are automata which, in addition to reversal-bounded
		counters, feature a stack. They can model recursive
		programs with numeric data types~\cite{DBLP:conf/cav/HagueL11}.
		Whether boundedness is decidable was stated as open
		in~\cite{DBLP:conf/dlt/EremondiIM15,DBLP:journals/ijfcs/EremondiIM18}.
	\item Are there languages of RBCA of intermediate growth? As far as we
		know, this is a long-standing open question in
		itself~\cite{DBLP:conf/stacs/IbarraR86}. The \emph{growth} of a
		language $L\subseteq\Sigma^*$ is the \emph{counting function} $g_L\colon
		\N\to\N$, where $g_L(n)$ is the number of words of length $n$
		in $L$. This concept is closely tied to boundedness: For
		regular and context-free languages, it is known that a language
		has polynomial growth if and only if it is bounded (and it has
		exponential growth otherwise). A language is said to have
		\emph{intermediate growth} if it has neither polynomial nor
		exponential growth.
\end{enumerate}

\vspace{-0.2cm}
\subsubsection{Contribution I:} We prove versions of one of the main results
in~\cite{DBLP:conf/icalp/CzerwinskiHZ18}, one for RBCA and one for PRBCA.
Specifically, the paper~\cite{DBLP:conf/icalp/CzerwinskiHZ18} not only shows
that boundedness is decidable for VASS, but it introduces a general class of
\emph{unboundedness predicates} for formal languages. It is then shown in
\cite{DBLP:conf/icalp/CzerwinskiHZ18} that any unboundedness predicate is
decidable for VASS if and only if it is decidable for regular languages.
Our first two main results are:
\begin{enumerate}[leftmargin=*, label=MR\arabic*:]
	\item Deciding any unboundedness predicate for RBCA reduces in $\NP$ to
		deciding the same predicate for regular languages.
	\item Deciding any unboundedness predicate for PRBCA reduces in $\NP$
		to deciding the same predicate for context-free languages.
\end{enumerate}
However, it should be noted that our results only apply to those unboundedness
predicates from~\cite{DBLP:conf/icalp/CzerwinskiHZ18} that are
\emph{one-dimensional}. Fortunately, these are enough for our applications.
These results allow us to settle questions (Q1)--(Q3) above and derive the
exact complexity of several other problems.  It follows that boundedness for
both RBCA and PRBCA is $\coNP$-complete, thus answering (Q1) and (Q2).
Furthermore, the proof shows that if boundedness of a PRBCA does not hold, then
its language has exponential growth. This implies that there are no RBCA
languages of intermediate growth (thus settling (Q3)), and even that the same
holds for PRBCA. In particular, deciding polynomial growth of (P)RBCA is
$\coNP$-complete and deciding exponential growth of (P)RBCA is $\NP$-complete.
We can also derive from our result that deciding whether a (P)RBCA language is
infinite is $\NP$-complete (but this also follows easily
from~\cite{DBLP:conf/cav/HagueL11}, see \cref{sec:results}). Finally, our results
imply that it is $\PSPACE$-complete to decide if an RBCA language
$L\subseteq\Sigma^*$ is \emph{factor universal}, meaning it contains every word of
$\Sigma^*$ as a factor (i.e.\ as an infix). Whether this problem is decidable
for RBCA was also left as an open problem
in~\cite{DBLP:conf/dlt/EremondiIM15,DBLP:journals/ijfcs/EremondiIM18} (under
the name \emph{infix density}).

We prove our results (MR1) and (MR2) by first translating (P)RBCA into models that have
$\Z$-counters instead of reversal-bounded counters. A \emph{$\Z$-counter} is
one that can be incremented and decremented, but cannot be tested for zero.
Moreover, it can assume negative values. With these counters, acceptance is
defined by reaching a configuration where all counters are zero (in
particular, the acceptance condition permits a single zero-test on each
counter). Here, finite automata with $\Z$-counters are called
\emph{$\Z$-VASS}~\cite{DBLP:conf/rp/HaaseH14}.
$\Z$-counters are also known as \emph{blind
counters}~\cite{DBLP:journals/tcs/Greibach78a} and it is a standard fact that
RBCA are equivalent (in terms of accepted languages) to $\Z$-VASS~\cite[Theorem
2]{DBLP:journals/tcs/Greibach78a}. 

Despite the equivalence between RBCA and $\Z$-VASS being so well-known,
there was apparently no known translation from RBCA to $\Z$-VASS in polynomial
time.  Here, the difficulty stems from simulating zero-tests (which can occur
an unbounded number of times in an RBCA): To simulate these,
the $\Z$-VASS needs to keep track of which counter has completed which
incrementing/decrementing phase, using only polynomially many control states.
It is also not obvious how to employ the $\Z$-counters for this, as they are
only checked in the end.

\vspace{-0.2cm}
\subsubsection{Contribution II:} As the first step of showing (MR1), we show that
\begin{enumerate}[leftmargin=*, resume*]
	\item RBCA can be translated (preserving the language) into $\Z$-VASS in logarithmic space.
\end{enumerate}
This also implies
that translations to and from another equivalent model, Parikh
automata~\cite{KlaedtkeRuess2003}, are possible in polynomial time: It was
recently shown that Parikh automata (which have received much attention in
recent
years~\cite{DBLP:conf/icalp/ClementeCLP17,DBLP:journals/mst/CadilhacKM18,DBLP:journals/ita/CadilhacFM12,DBLP:journals/ijfcs/CadilhacFM12,DBLP:conf/icalp/BostanCKN20,DBLP:conf/fsttcs/FiliotGM19})
can be translated in polynomial time into $\Z$-VASS~\cite{HaaseZ19}. Together
with our new result, this implies that one can translate among RBCA, $\Z$-VASS,
and Parikh automata in polynomial time.  Furthermore, our result yields a
logspace translation of PRBCA into $\Z$-grammars, an extension of
context-free grammars with $\Z$-counters. The latter is the first step for
(MR2).

%
%
%

%% file: results.tex

\begin{table}[t]
\begin{center}
\begin{tabular}{lll}\toprule
Problem &  $\ZVASS$/$\RBCA$ & $\Z$-grammars/$\PRBCA$ \\ \midrule
Boundedness		& $\coNP$-complete & $\coNP$-complete \\
Finiteness 		& $\coNP$-complete & $\coNP$-complete \\
Factor universality~	& $\PSPACE$-complete~  & undecidable   \\\bottomrule
\end{tabular}
\end{center}
\caption{Complexity results. The completeness statements are meant with respect to deterministic logspace reductions.}\label{table-complexity}
\end{table}

\subsubsection{Reversal-bounded counter automata and pushdowns.}
A \emph{pushdown automaton with $k$ counters} is a tuple $\cA = (Q,\Sigma,\Gamma,q_0,T,F)$
where $Q$ is a finite set of states, $\Sigma$ is an input alphabet, $\Gamma$ is a stack alphabet,
$q_0 \in Q$ is an initial state,
$T$ is a finite set of transitions $(p,w,\mathrm{op},q) \in Q \times \Sigma^* \times \mathrm{Op} \times Q$,
and $F\subseteq Q$ is a set of final states.
Here $\mathrm{Op}$ is defined as
\[
	\mathrm{Op} = \{ \mathsf{inc}_i, \mathsf{dec}_i, \mathsf{zero}_i, \mathsf{nz}_i  \mid 1 \le i \le k \} \cup \Gamma \cup \bar \Gamma \cup \{\varepsilon\},
\]
containing counter and stack operations.
Here $\bar \Gamma = \{ \bar \gamma \mid \gamma \in \Gamma \}$ is a disjoint copy of $\Gamma$.
A \emph{configuration} is a tuple $(p,\alpha,\vec v) \in Q \times \Gamma^* \times \N^k$.
We write $(p,\alpha,\vec u)\xrightarrow{w}(p',\alpha',\vec u')$ if there is a
$(p,w,\mathrm{op},p')\in T$ such that one of the following holds:
\begin{itemize}
\item $\mathrm{op} = \mathsf{inc}_i$, $\vec u' = \vec u + \vec e_i$, and $\alpha' = \alpha$ where $\vec e_i \in \N^k$ is the $i$-th unit vector,
\item $\mathrm{op} = \mathsf{dec}_i$, $\vec u' = \vec u - \vec e_i$, and $\alpha' = \alpha$,
\item $\mathrm{op} = \mathsf{zero}_i$, $\vec u[i] = 0$, $\vec u' = \vec u$, and $\alpha' = \alpha$
\item $\mathrm{op} = \mathsf{nz}_i$, $\vec u[i] \neq 0$, $\vec u' = \vec u$, and $\alpha' = \alpha$,
\item $\mathrm{op} = \gamma \in \Gamma$, $\vec u' = \vec u$, and $\alpha' = \alpha \gamma$,
\item $\mathrm{op} = \bar \gamma \in \bar \Gamma$, $\vec u' = \vec u$, and $\alpha' \gamma = \alpha$,
\item $\mathrm{op} = \varepsilon$, $\vec u' = \vec u$, and $\alpha' = \alpha$.
\end{itemize}
We extend this notation to longer runs in the natural way.

A \emph{$(k,r)$-PRBCA} (pushdown reversal-bounded counter automaton) $(\cA,r)$
consists of a pushdown automaton with $k$ counters $\cA$ and a number $r \in \N$, encoded in unary.
A counter $\ctr_i$ \emph{reverses} if the last (non-test) operation affecting it was $\mathsf{inc}_i$ and the next operation is $\mathsf{dec}_i$, or vice versa.
A run is \emph{$r$-reversal bounded} if every counter reverses at most $r$ times.
The \emph{language} of $(\cA,r)$ is 
\[
  \begin{aligned}
	L(\cA,r)=\{w\in\Sigma^* \mid &~\exists q\in F,~r\text{-reversal bounded run } (q_0,\varepsilon,\vec 0)\xrightarrow{w}(q,\varepsilon,\vec 0) \}.
  \end{aligned}
\]

A \emph{$(k,r)$-RBCA} (reversal-bounded counter automaton) is a $(k,r)$-PRBCA where $\cA$ only uses counter operations.
We denote by $\RBCA$ and $\PRBCA$ the class of (P)RBCA languages.

Notice that we impose the reversal bound \emph{externally} (following
\cite{DBLP:conf/cav/HagueL11}) whereas in alternative definitions found in the
literature the automaton has to ensure \emph{internally} that the number of
reversals on every (accepting) run does not exceed $r$,
e.g.~\cite{DBLP:journals/jacm/Ibarra78}.  Clearly, our definition subsumes the
latter one; in particular, \cref{main-conversion} also holds for (P)RBCAs with
an internally checked reversal bound.

A $d$-dimensional \emph{$\Z$-VASS ($\Z$-vector addition system with states)} is a
tuple $\cV=(Q,\Sigma,q_0,T,F)$, where
$Q$ is a finite set of states, $\Sigma$ is an alphabet,
$q_0 \in Q$ is an initial state, 
$T$ is a finite set of transitions $(p,w,\vec v,p') \in Q \times \Sigma^* \times \Z^d \times Q$,
and $F\subseteq Q$ is a set of final states.
A \emph{configuration} of a $\Z$-VASS is a tuple $(p,\vec v) \in Q \times \Z^d$.
We write $(p,\vec u)\xrightarrow{w}(p',\vec u')$ if there is a
transition $(p,w,\vec v,p')$ such that $\vec u'=\vec u+\vec v$.
We extend this notation to longer runs in the natural way.
The \emph{language} of the $\Z$-VASS
is defined as
\[ L(\cV)=\{w\in\Sigma^* \mid \exists q\in F\colon~(q_0,\vec 0)\xrightarrow{w}(q,\vec 0)\}. \]
A \emph{($d$-dimensional) $\Z$-grammar} is a tuple $G = (N,\Sigma,S,P)$ with
disjoint finite sets $N$ and $\Sigma$ of nonterminal and terminal symbols, a
start nonterminal $S \in N$, and a finite set of productions $P$ of the form
$(A,u,\vec v) \in N \times (N \cup \Sigma)^* \times \Z^d$. We also write $(A\to
u,\bv)$ instead of $(A,u,\bv)$.  We call $\vec v$ the \emph{(counter) effect}
of the production $(A \to u,\bv)$.  For words $x,y\in(N\cup \Sigma)^{*}$, we
write $x\derivv{\bv} y$ if there is a production $(A\to u,\bv)$ such that
$x=rAs$ and $y=rus$. Moreover, we write $x\derivvs{\bv} y$ if there are words
$x_1,\ldots,x_n\in(N\cup \Sigma)^*$ and $\bv_1,\ldots,\bv_n\in\Z^d$ with
$x\derivv{\bv_1} x_1\derivv{\bv_2}\cdots \derivv{\bv_n}x_n=y$ and
$\bv=\bv_1+\cdots+\bv_n$. We use the notation $\deriv$ if the counter effects
do not matter: We have $x\deriv y$ if there exists $\bv$ such that
$x\deriv[\bv] y$; and similarly for $\derivs$. If derivations are restricted
to a subset $Q\subseteq P$ of productions, we write $\deriv[Q]$ (resp.\
$\derivs[Q]$).

The \emph{language} of the $\Z$-grammar $G$ is the set of all words $w \in
\Sigma^*$ such that $S\derivvs{\bzero} w$. In other words, if there exists a
derivation $S \derivs w$ where the effects of all occurring productions sum to
the zero vector $\vec 0$. $\Z$-grammars of dimension $d$ are also known as
\emph{valence grammars over $\Z^d$}~\cite{DBLP:journals/tcs/FernauS02}.

For our purposes it suffices to assume a unary encoding of the $\Z^d$-vectors (effects) occurring in $\Z$-VASS and $\Z$-grammars.
However, this is not a restriction: Counter updates with $n$-bit binary encoded numbers
can be easily simulated with unary encodings at the expense of $dn$ many fresh counters
(see \cref{appendix-counter-updates}).

\subsubsection{Conversion results.}
The following is our first main theorem:
\begin{theorem}\label{main-conversion}
	RBCA can be converted into $\Z$-VASS in logarithmic space. \linebreak 
	PRBCA can be converted into $\Z$-grammars in logarithmic
	space.
\end{theorem}
By \emph{convert}, we mean a translation that preserves the accepted
(resp.\ generated) language.
There are several machine models that are equivalent (in terms of accepted
languages) to RBCA. With \cref{main-conversion}, we provide the last missing
translation:
\begin{corollary}
	The following models can be converted into each other in logarithmic
	space: (i)~RBCA, (ii)~$\Z$-VASS, (iii)~Parikh automata with $\exists$PA
	acceptance, and (iv)~Parikh automata with semilinear acceptance.
\end{corollary}
\newcommand{\conv}[2]{(#1)$\Rightarrow$(#2)}
Roughly speaking, a Parikh automaton is a machine with counters that can only
be incremented. Then, a run is accepting if the final counter values belong to
some semilinear set. Parikh automata were introduced by Klaedtke and
Rue\ss~\cite{KlaedtkeRuess2003}, where the acceptance condition is specified
using a semilinear representation (with base and period vectors), yielding~(iv)
above. As done, e.g., in~\cite{DBLP:conf/lics/HalfonSZ17}, one could also
specify it using an existential Presburger formula (briefly $\exists$PA),
yielding the model in~(iii) above. \Cref{main-conversion} proves \conv{i}{ii},
whereas \conv{ii}{i} is easy (a clever and very efficient translation is given
in~\cite[Theorem 4.5]{DBLP:journals/iandc/JantzenK03}).  Moreover,
\conv{ii}{iii} and \conv{ii}{iv} are clear as well. For \conv{iii}{ii}, one can
proceed as in~\cite[Prop.~V.1]{HaaseZ19}, and \conv{iv}{ii} is also simple.

\subsubsection{Unboundedness predicates.}
We shall use \cref{main-conversion} to prove our second main theorem, which
involves unboundedness predicates as introduced
in~\cite{DBLP:conf/icalp/CzerwinskiHZ18}.
In~\cite{DBLP:conf/icalp/CzerwinskiHZ18}, unboundedness predicates can be
one-dimensional or multi-dimensional, but in this work, we only consider
one-dimensional unboundedness predicates.

Let $\Sigma$ be an alphabet. A \emph{(language) predicate} is a set of
languages over $\Sigma$. If $\pP$ is a predicate and $L\subseteq\Sigma^*$ is a
language, then we write $\pP(L)$ to denote that $\pP$ holds for the language
$L$ (i.e.\ $L \in \pP$). A predicate $\pP$ is called a \emph{(one-dimensional) unboundedness
predicate} if  the following conditions are met for all $K,L\subseteq\Sigma^*$:
\vspace{0.1cm}

\begin{minipage}{0.5\textwidth}
\begin{enumerate}[leftmargin=*, label=(U\arabic*)]
	\item If $\pP(K)$ and $K\subseteq L$, then $\pP(L)$.
	\item If $\pP(K\cup L)$, then $\pP(K)$ or $\pP(L)$.
\end{enumerate}
\end{minipage}%
\begin{minipage}{0.5\textwidth}
	\begin{enumerate}[leftmargin=*, label=(U\arabic*), start=3]
	\item If $\pP(K \cdot L)$, then $\pP(K)$ or $\pP(L)$.
	\item $\pP(L)$ if and only if $\pP(F(L))$.
	\end{enumerate}
\end{minipage}
\vspace{0.1cm}

\noindent
Here $F(L)=\{v\in\Sigma^* \mid \exists u,w\in\Sigma^*\colon uvw\in L\}$ is the
set of \emph{factors} of $L$ (sometimes also called \emph{infixes}).
In particular, the last condition says that $\pP$
only depends on the set of factors occurring in a language.

For an unboundedness predicate $\pP$ and a class $\cC$ of finitely represented languages (such as automata or grammars),
let $\pP(\cC)$ denote the problem of
deciding $\pP$ for a given language $L$ from $\cC$. Formally, $\pP(\cC)$ is the following decision problem:
\begin{description}
	\item[Given] A language $L$ from $\cC$.
	\item[Question] Does $\pP(L)$ hold?
\end{description}
For example, $\pP(\RBCA)$ is the problem of deciding $\pP$ for reversal-bounded
multi-counter automata and $\pP(\NFA)$ is the problem of deciding $\pP$ for
NFAs.  We mention that the axioms (U1)--(U4) are slightly stronger than the
axioms used in~\cite{DBLP:conf/icalp/CzerwinskiHZ18}, but the resulting set of
decision problems is the same with either definition (since in
\cite{DBLP:conf/icalp/CzerwinskiHZ18}, one always decides whether $\pP(F(L))$ holds).
Thus, the statement of \linebreak
\cref{main-unboundedness} is unaffected by which
definition is used. See \cref{appendix-unboundedness-predicates} for details.

The following examples of (one-dimensional) unboundedness predicates for
languages $L\subseteq\Sigma^*$ have already been established in
\cite{DBLP:conf/icalp/CzerwinskiHZ18}. We mention them here to give an
intuition for the range of applications of our results:
\begin{description}
\item[Not being bounded] Let $\pP_{\mathsf{notb}}(L)$ if and only if $L$ is
\emph{not} a bounded language. 
\item[Non-emptiness]  Let $\pP_{\ne\emptyset}(L)$ if and only if $L\ne\emptyset$.
\item[Infinity] Let $\pP_{\infty}(L)$ if and only if $L$ is infinite.
\item[Factor-universality] Let $\pP_{\mathsf{funi}}(L)$ if and only if $\Sigma^*\subseteq F(L)$.
\end{description}
It is not difficult to prove that these are unboundedness predicates, but
proofs can be found in \cite{DBLP:conf/icalp/CzerwinskiHZ18}.
The following is our second main theorem:
\begin{theorem}\label{main-unboundedness}
Let $\pP$ be a one-dimensional unboundedness predicate.  There is an $\NP$
reduction from $\pP(\PRBCA)$ to $\pP(\PDA)$. Moreover, there is an $\NP$
reduction from $\pP(\RBCA)$ to $\pP(\NFA)$.
\end{theorem}
Here, an \emph{$\NP$ reduction} from problem $A\subseteq\Sigma^*$ to
$B\subseteq\Sigma^*$ is a non-deterministic polynomial-time Turing machine such
that for every input word $w\in\Sigma^*$, we have $w \in A$ iff there exists a
run of the Turing machine producing a word in $B$.

Let us now see some applications of \cref{main-unboundedness}, see also \cref{table-complexity}. The following completeness results are all meant w.r.t.\ deterministic logspace reductions.
\begin{corollary}\label{complexity-boundedness}
Boundedness for $\PRBCA$ and for $\RBCA$ is $\coNP$-complete. 
\end{corollary}
For \cref{complexity-boundedness}, we argue that deciding
\emph{non-boundedness} is $\NP$-complete. To this end, we apply
\cref{main-unboundedness} to the predicate $\pP_{\mathsf{notb}}$ and obtain an
$\NP$ upper bound, because boundedness for context-free languages is decidable
in polynomial time~\cite{DBLP:journals/ijfcs/GawrychowskiKRS10}.  The $\NP$
lower bound follows easily from $\NP$-hardness of the non-emptiness problem for
$\RBCA$~\cite[Theorem~3]{GURARI1981220} and thus $\PRBCA$.

\begin{corollary}\label{complexity-finiteness}
Finiteness for $\PRBCA$ and for $\RBCA$ is $\coNP$-complete.
\end{corollary}
\newcommand{\EPA}{$\exists$PA}
We show \cref{complexity-finiteness} by proving that checking
\emph{infinity} is $\NP$-complete. The upper bound follows from
\cref{main-unboundedness} via the predicate $\pP_\infty$. As above,
$\NP$-hardness is inherited from the non-emptiness problem for $\RBCA$ and $\PRBCA$.

The results in \cref{complexity-finiteness} are, however, not new. They follow 
directly from the fact that for a given $\PRBCA$ (or $\RBCA$), one can
construct in polynomial time a formula in existential Presburger arithmetic (\EPA) for its
Parikh image, as shown in \cite{DBLP:journals/jacm/Ibarra78} for RBCA and in
\cite{DBLP:conf/cav/HagueL11} for PRBCA. It is a standard result about
\EPA\ that for each formula $\varphi$, there exists a
bound $B$ such that (i)~$B$ is at most exponential in the size of $\varphi$ and
(ii)~$\varphi$ defines an infinite set if and only if $\varphi$ is satisfied
for some vector with some entry above $B$. For example, this can be deduced
from~\cite{DBLP:conf/rta/Pottier91}.  Therefore, one can easily construct a
second \EPA\ formula $\varphi'$ such that $\varphi$ defines an
infinite set if and only if $\varphi'$ is satisfiable.

\begin{corollary}\label{complexity-factor-universality}
	Factor universality for $\RBCA$ is $\PSPACE$-complete.
\end{corollary}
Whether factor universality is decidable for $\RBCA$ was left as an open
problem in~\cite{DBLP:conf/dlt/EremondiIM15,DBLP:journals/ijfcs/EremondiIM18}
(there under the term \emph{infix density}).
\Cref{complexity-factor-universality} follows from \cref{main-unboundedness} using $\pP_{\mathsf{funi}}$, because
factor universality for NFAs is $\PSPACE$-complete: To decide if
$\Sigma^*\subseteq F(R)$, for a regular language $R$, we can just compute an
automaton for $F(R)$ and check inclusion in $\PSPACE$. For the lower bound, one
can reduce the $\PSPACE$-complete universality problem for NFAs, since for
$R\subseteq\Sigma^*$, the language $(R\#)^*\subseteq(\Sigma\cup\{\#\})^*$ is factor universal if
only if $R=\Sigma^*$. Note that factor universality is known to be undecidable
already for one-counter
languages~\cite{DBLP:journals/ijfcs/EremondiIM18}, and thus in particular for
$\PRBCA$.  However, it is decidable for pushdown automata with a bounded number
of reversals of the stack~\cite{DBLP:journals/ijfcs/EremondiIM18}.

\vspace{-0.2cm}
\subsubsection{Beyond pushdowns.} \Cref{main-unboundedness} raises the question of
whether for any class $\cM$ of machines, one can reduce any unboundedness
predicates for $\cM$ extended with reversal-bounded counters to the same predicate for just $\cM$. This is not the
case: For example, consider second-order pushdown automata, short $2$-PDA. If
we extend these by adding reversal-bounded counters, then we obtain $2$-PRBCA.
Then, the infinity problem is decidable for $2$-PDA~\cite{Hayashi1973}
(see~\cite{DBLP:conf/icalp/Zetzsche15,DBLP:conf/popl/HagueKO16,DBLP:conf/lics/ClementePSW16,DBLP:conf/icalp/BarozziniCCP20,DBLP:conf/icalp/BarozziniPW22,DBLP:conf/fsttcs/Parys17}
for stronger results). However, the class of $2$-PRBCA does not even have
decidable emptiness, let alone decidable infinity. This is shown
in~\cite[Proposition~7]{DBLP:journals/corr/Zetzsche15} (see~\cite[Theorem
4]{DBLP:journals/tcs/Kobayashi19} for an alternative proof).  Thus, infinity
for $2$-PRBCA cannot be reduced to infinity for $2$-PDA.

\vspace{-0.2cm}
\subsubsection{Growth.}
Finally, we employ the methods of the proof of \cref{main-unboundedness} to
show a dichotomy of the growth behavior of languages accepted by RBCA.  For an
alphabet $\Sigma$, we denote by $\Sigma^{\le m}$ the set of all words over
$\Sigma$ of length at most $m$. We say that a language $L\subseteq\Sigma^*$ has
\emph{polynomial
growth}\footnote{In~\cite{DBLP:journals/ijfcs/GawrychowskiKRS10}, polynomial
and exponential growth are defined with $\Sigma^m$ in place of $\Sigma^{\le
m}$, but this leads to equivalent notions, see \cref{appendix-growth}} if there
is a polynomial $p(x)$ such that $|L\cap \Sigma^{\le m}|\le p(m)$ for all $m\ge
0$.  Languages of polynomial growth are also called \emph{sparse} or
\emph{poly-slender}.  We say that $L$ has \emph{exponential growth} if there is
a real number $r>1$ such hat $|L\cap \Sigma^{\le m}|\ge r^m$ for infinitely
many $m$.  Since a language of the form $w_1^*\cdots w_n^*$ clearly has
polynomial growth, it is well-known that bounded languages have polynomial
growth. We show that (a)~within the PRBCA languages (and in particular within
the RBCA languages), the converse is true as well and (b)~all other languages
have exponential growth (in contrast to some models, such as
2-PDA~\cite{grigorchuk1999example}, where this dichotomy does not hold):
\begin{theorem}\label{main-growth}
	Let $L$ be a language accepted by a PRBCA. Then $L$ has polynomial
growth if and only if $L$ is bounded. If $L$ is not bounded, it has exponential
growth.
\end{theorem}

%% file: rbc-zvass-conversion.tex


\subsubsection{Reducing the number of reversals to one.} 
In this section we prove \cref{main-conversion}, the conversion from RBCA to $\Z$-VASS.
In \cite[Lemma 1]{GURARI1981220}, it is claimed that given a
$(k,r)$-RBCA, one can construct in time polynomial in $k$ and $r$ a
$(k\lceil (r+1)/2\rceil,1)$-RBCA that accepts the same language.
The reference \cite{baker1974reversal} that they provide does include
such a construction \cite[proof of Theorem 5]{baker1974reversal}.  The construction in~\cite{baker1974reversal} is only a rough sketch and makes no claims about complexity, but by our reading of the construction, it keeps track of the reversals of each counter in the state, which would result in an exponential blow-up.

Instead, we proceed as follows. Consider a $(k,r)$-RBCA with counters
\linebreak
$\ctr_1, \dots, \ctr_k$.
Without loss of generality, assume $r=2m-1$.
We will construct an equivalent $(2k(r+1),1)$-RBCA.
Looking at the behavior of a single counter $\ctr_i$,
we can decompose every $r$-reversal bounded run
into subruns without reversals.
We call these subruns \emph{phases} and number them from $1$ to at most $2m$.
The odd (even) numbered phases are \emph{positive} (\emph{negative}), where $\ctr_i$ is only incremented (decremented).
We replace $\ctr_i$ by $m$ one-reversal counters $\ctr_{i,1}, \dots, \ctr_{i,m}$,
where $\ctr_{i,j}$ records the increments on $\ctr_i$ during the positive phase $2j-1$.

However, our machine needs to keep track of which counters are 
in which phase, in order to know which of the counters $\ctr_{i,j}$ it currently
has to use. We achieve this as follows: For each of the $k$
counters $\ctr_i$, we also have an additional set of $2m = r + 1$ ``phase counters'' $\pctr_{i,1}, \dots, \pctr_{i,2m}$ to
store which phase we are in. This gives $km+k(r+1)\le 2k(r+1)$
counters in total. We encode that counter $\ctr_i$ is in phase $j$ by setting $\pctr_{i,j}$ to $1$ and setting $\pctr_{i,j'}$ to $0$ for each $j' \neq j$.
Since we only ever increase the
phase, the phase counters are one-reversal as well.

Using non-zero-tests, at any point, the automaton can nondeterministically guess and verify the current phase of each counter.
This allows it to pick the correct counter $\ctr_{i,j}$ for each instruction.
When counter $\ctr_i$ is in a positive phase $2j-1$, then increments and decrements on $\ctr_i$ are simulated as follows:
\begin{description}
\item[increment] increment $\ctr_{i,j}$
\item[decrement] go into the next (negative) phase $2j$; then
  non-deterministically pick some $\ell\in[1,j]$ and decrement $\ctr_{i,\ell}$. We 
  cannot simply decrement $\ctr_{i,j}$ as we might have switched to phase $j$ while $\ctr_i$ had a non-zero value and hence it is possible that $\ctr_i$ could be decremented further than just $\ctr_{i,j}$ allows.
\end{description}

When counter $\ctr_i$ is in a negative phase $2j$, then we
simulate increments and decrements as follows:
\begin{description}
\item[increment] go into the next phase $2j+1$ (unless $j = m$; then the machine blocks) and increment $\ctr_{i,j+1}$.
\item[decrement] non-deterministically pick some $\ell\in[1,j]$ and decrement $\ctr_{i,\ell}$.
\end{description}
Finally, to simulate a zero-test on $\ctr_i$, we test all counters $\ctr_{i,1},\ldots,\ctr_{i,m}$ for zero,
while for the simulation of a non-zero-test on $\ctr_i$ we non-deterministically
pick one of the counters $\ctr_{i,1},\ldots,\ctr_{i,m}$ to test for non-zero.

Correctness can be easily verified by the following properties. If at some point $\ctr_i$ is in phase $2j-1$ or $2j$ then
(i) $\sum_{\ell=1}^{j} \ctr_{i,\ell} = \ctr_i$,
(ii) the counters $\ctr_{i,1}, \dots, \ctr_{i,j}$ have made at most one reversal, and
(iii) the counters $\ctr_{i,j+1}, \dots, \ctr_{i,m}$ have not been touched (in particular, they are zero).
Furthermore, if $\ctr_i$ is in a positive phase $2j-1$ then $\ctr_{i,j}$ has made no reversal yet.

Note that this construction replaces every transition of the original system with $\bigO(r)$ new transitions (and states). Our construction therefore yields only a linear blowup in the size of the system (constant if $r$ is fixed). See \cref{app:1-reversal-construction} for the details of the construction.

\subsubsection{From $1$-reversal to $\Z$-counters.}
We now turn the $(k,1)$-RBCA into a $\Z$-VASS.
The difference between a 1-reversal-bounded counter and a $\Z$-counter is that
(i) a non-negative counter should block if it is decremented on counter value $0$, and
(ii) a 1-reversal-bounded counter allows (non-)zero-tests.
Observe that all zero-tests occur before the first increment or after the last decrement.
All non-zero-tests occur between the first increment and the last decrement.

If the number $k$ of counters is bounded, then the following simple solution works.
The $\Z$-VASS stores the information which of the counters has not been incremented yet
and which counters will not be incremented again in the future.
This information suffices to simulate the counters faithfully (in terms of the properties (i) and (ii) above)
and increases the state space by a factor of $2^k \cdot 2^k$.
The latter information needs to be guessed (by the automaton) and is verified by means that all counters are zero in the end.

In the general case we introduce a variant of $\Z$-VASS that can guess polynomially many bits in the beginning and read them throughout the run.
A $d$-dimensional \emph{$\Z$-VASS with guessing ($\Z$-VASSG)}
has almost the same format as a $d$-dimensional $\Z$-VASS,
except that each transition additionally carries a propositional formula
over some finite set of variables $X$.
A word $w \in \Sigma^*$ is accepted by the $\Z$-VASSG if there exists an assignment
$\nu \colon X \to \{0,1\}$ and an accepting run $(q_0,\vec 0) \xrightarrow{w} (q,\vec 0)$
for some $q \in F$ such that all formulas appearing throughout the run are satisfied by $\nu$.


%

We have to eliminate zero- and non-zero-tests of the $(k,1)$-RBCA.
Whether a (non-)zero-test is successful depends on
which phase a counter is currently in (and whether in the end, every counter is zero; but we assume that our acceptance condition ensures this).
Each counter goes through at most $4$ phases:
\vspace{0.05cm}

\begin{minipage}{0.5\textwidth}
\begin{enumerate}
\item before the first increment,
\item the ``increment phase'',
\end{enumerate}
\end{minipage}%
\begin{minipage}{0.5\textwidth}
\begin{enumerate}[start=3]
\item the ``decrement phase'', and
\item after the last decrement.
\end{enumerate}
\end{minipage}
\vspace{0.05cm}

\noindent
Hence, every run can be decomposed into $4k$ (possibly empty) segments, in which no
counter changes its phase.
The idea is to guess the phase of each counter in each
segment. Hence, we have propositional variables $\var{p}_{i,j,\ell}$ for $i\in[1,4k]$,
$j\in[1,k]$, and $\ell\in[1,4]$. Then $\var{p}_{i,j,\ell}$ is true iff in
segment $i$, counter $j$ is in phase $\ell$. We will have to check that
the assignment is \emph{admissible} for each counter, meaning that the sequence
of phases for each counter adheres to the order described above.

We modify the machine as follows. In its state, it keeps a number
$i\in[1,4k]$ which holds the current segment. At the beginning of the
run, the machine checks that the assignment $\nu$ is admissible using
a propositional formula: It checks that (i) for each segment $i$ and
each counter $j$ there exists exactly one phase $\ell$ so that $\var{p}_{i,j,\ell}$ is true,
and (ii) the order of phases above is obeyed.
Then, for every operation on a counter, the
machine checks that the operation is consistent with the current
segment. Moreover, if the current operation warrants a change of the
segment, then the segment counter $i$ is incremented.
For example, if a counter in phase 1 is incremented,
it switches to phase 2 and the segment counter is incremented;
or, if a counter in phase 3 is tested for zero,
it switches to phase 4 and the segment counter is incremented.

With these modifications, we can zero-test by checking
variables corresponding to the current segment: A zero-test can only succeed in phase 1 and 4. Similarly, for a non-zero-test, we can check if the counter is in phase 2 or 3.

\subsubsection{Turning a $\Z$-VASSG into a $\Z$-VASS.}

To handle the general case mentioned above,
we need to show how to convert $\Z$-VASSG into ordinary $\Z$-VASS.
In a preparatory step, we ensure that each formula is a literal.
A transition labeled by a formula $\varphi$ is replaced by a series-parallel graph:
After bringing $\varphi$ in negation normal form by pushing negations inwards,
we can replace conjunctions by a series composition
and disjunctions by a parallel composition (non-determinism).

The $\Z$-VASS works as follows. In addition to the original counters of the
$\Z$-VASSG, it has for each variable $\var{x}\in X$ two additional
counters: $\var{x}^+$ and $\var{x}^-$. Here,
$\var{x}^+$ ($\var{x}^-$) counts how many times $\var{x}$ is read with a
positive (negative) assignment.  By making sure that either
$\var{x}^+=0$ or $\var{x}^-=0$ in the end, we guarantee
that we always read the same value of $\var{x}$.

Thus, in order to check a literal, our $\Z$-VASS increments the
corresponding counter. In the end, before reaching a final state, it
goes through each variable $\var{x}\in X$ and either enters a loop
decrementing $\var{x}^+$ or a loop decrementing $\var{x}^-$. Then, it
can reach the zero vector only if all variable checks had been
consistent.

\subsubsection{From PRBCA to $\Z$-grammars.}
It remains to convert in logspace an $(r,k)$-PRBCA into an equivalent $\Z$-grammar.
Just as for converting an RBCA into a $\Z$-VASS,
one can convert a PRBCA into an equivalent \emph{$\Z$-PVASS} (pushdown vector addition system with $\Z$-counters).
Afterwards, one applies the classical transformation from pushdown automata to context-free grammars
(a.k.a.\ triple construction), cf.~\cite[Lemma~2.26]{DBLP:books/lib/AhoU72}:
We introduce for every state pair $(p,q)$ a nonterminal $X_{p,q}$, deriving all words which are read between $p$ to $q$ (starting and ending with empty stacks).
For example, we introduce productions $X_{p,q} \to a X_{p',q'} b$ for all push transitions
$(p, a, \gamma, p')$
and pop transitions
$(q', b, \bar \gamma, q)$.
The counter effects of transitions in the $\Z$-PVASS (vectors in $\Z^k$) are translated into effects of the productions,
e.g.\ the effect of the production $X_{p,q} \to a X_{p',q'} b$ above is the sum of the effects of the corresponding push- and pop-transition. 

%% file: unboundedness-reduction.tex

\subsubsection{Proof overview.} 
In this section, we prove \cref{main-unboundedness}. Let us begin with a
sketch. Our task is to take a PRBCA $\cA$ and
non-deterministically compute a PDA $\cA'$ so that $L(\cA)$ satisfies $\pP$ if
and only if some of the outcomes for $\cA'$ satisfy $\pP$. It will be clear
from the construction that if the input was an RBCA, then the resulting
PDA will be an NFA.
Using \cref{main-conversion} we will phrase the main part of the reduction
in terms of $\Z$-grammars, meaning we take a $\Z$-grammar $G$ as input and
non-deterministically compute context-free grammars $G'$.

The idea of the reduction is to identify a set of productions in $G$ that, in
some appropriate sense, can be canceled (regarding the integer counter values)
by a collection of other productions. Then, $G'$ is obtained by only using a
set of productions that can be canceled. Moreover, these productions are used
regardless of what counter updates they perform. Then, to show the correctness,
we argue in two directions: First, we show that any word derivable by $G'$ occurs
as a factor of $L(G)$. Essentially, this is because each production used in
$G'$ can be canceled by adding more productions in $G$, thus yielding a
complete derivation of $G$. Thus, we have that $L(G')\subseteq F(L(G))$, which
by the axioms of unboundedness predicates means that $\pP(L(G'))$ implies
$\pP(L(G))$.  Second, we show that $L(G)$ is a finite union of products
(i.e.\ concatenations) $P_i = L_1 \cdot L_2 \cdots L_k$
such that each $L_i$ is either finite or included in
$L(G')$ for some $G'$ among all non-deterministic outcomes.
Again, by the axioms of unboundedness
predicates, this means that if $\pP(L(G))$, then $\pP(L(G'))$ must hold for
some $G'$.

\subsubsection{Unboundedness predicates and finite languages.} Before we start
with the proof, let us observe that we may assume that our unboundedness predicate
is only satisfied for infinite sets. 
First, suppose $\pP$ is satisfied for $\{\varepsilon\}$. This implies that
$\pP=\pP_{\ne\emptyset}$ and hence we can just decide whether $\pP(L)$ by
deciding whether $L\ne\emptyset$, which can be done in
$\NP$~\cite{DBLP:conf/cav/HagueL11}.  From now on, suppose that $\pP$ is not
satisfied for $\{\varepsilon\}$.  Consider the alphabet
$\Sigma_1:=\{a\in\Sigma \mid \pP(\{a\})\}$.
Now observe that if $K\subseteq\Sigma^*$ is finite, then by the axioms of
unboundedness predicates, we have $\pP(K)$ if and only if some letter from
$\Sigma_1$ appears in $K$. Thus, if $L\subseteq(\Sigma\setminus\Sigma_1)^*$,
then $\pP(L)$ can only hold if $L$ is infinite. This motivates the following
definition. Given a language $L\subseteq\Sigma^*$, we define
\[ \begin{aligned} L_0=L\cap (\Sigma\setminus\Sigma_1)^*, && L_1=L\cap \Sigma^*\Sigma_1\Sigma^*. \end{aligned} \]
Then, $\pP(L)$ if and only if $\pP(L_0)$ or $\pP(L_1)$. Moreover,
$\pP(L_1)$ is equivalent to $L_1\ne\emptyset$.

Therefore, our reduction proceeds as follows.  We construct (P)RBCA for $L_0$
and for $L_1$. This can be done in logspace, because intersections with regular
languages can be done with a simple product construction. Then, we check in
$\NP$ whether $L_1\ne\emptyset$. If yes, then we return ``unbounded''. If no, we
regard $\pP$ as an unboundedness predicate on languages over
$\Sigma\setminus\Sigma_1$ with the additional property that $\pP$ is only
satisfied for infinite languages. Thus, it suffices to prove \cref{main-unboundedness} in the case that $\pP$ is only
satisfied for infinite sets.

\subsubsection{Pumps and cancelation.}
In order to define our notion of cancelable productions, we need some terminology.
We will need to argue about derivation trees for $\Z$-grammars.
For any alphabet $\Gamma$ and $d\in\N$, let $\cT_{\Gamma,d}$ be the set of
all finite trees where every node is labeled by both (i)~a letter from $\Gamma$
and (ii)~a vector from $\Z^d$.
Suppose $G=(N,\Sigma,P,S)$ is a $d$-dimensional $\Z$-grammar. For a production
$p = (A \to u,\bv)$, we write $\counters(p):=\bv$ for its associated counter
effect.  To each derivation in $G$, we associate a derivation tree from $\cT_{N \cup \Sigma,d}$ as for
context-free grammars.  The only difference is that whenever we apply a
production $(A\to u,\bv)$, then the node corresponding to the rewritten $A$ is
also labeled with $\bv$. As in context-free grammars, the leaf nodes carry terminal
letters; their vector label is just $\bzero\in\Z^d$.

We extend the map $\counters$ to both vectors in $\N^P$ and to derivation trees.
If $\vec u\in\N^P$, then $\counters(\vec u) = \sum_{p \in P}\counters(p) \cdot \vec{u}[p]$.
Similarly, if $\tau$ is a derivation tree, then $\counters(\tau)\in\Z^d$ is the sum of all
labels from $\Z^d$.  A derivation tree $\tau$ for a derivation $A \derivs u$ is
called \emph{complete} if $A = S$, $u \in \Sigma^*$ and $\counters(\tau) =
\vec{0}$. In other words, $\tau$ derives a terminal word and the total counter
effect of the derivation is zero.  For such a complete derivation, we also
write $\mathsf{yield}(\tau)$ for the word $u$.  A derivation tree $\tau$ is
called a \emph{pump} if it is the derivation tree of a derivation of the form
$A\derivs uAv$ for some $u,v\in \Sigma^*$ and $A\in N$.  A subset $M\subseteq
N$ of the non-terminals is called \emph{realizable} if there exists a complete
derivation  of $G$ that contains all non-terminals in $M$ and no non-terminals
outside of $M$. 

A production $p$ in $P$ is called \emph{$M$-cancelable} if there exist pumps
$\tau_1,\ldots,\tau_k$ (for some $k\in\N$) such that (i)~$p$ occurs in some
$\tau_i$ and (ii)~$\counters(\tau_1)+\cdots+\counters(\tau_k)=\bzero$, i.e.\
the total counter effect of $\tau_1,\ldots,\tau_k$ is zero and (iii)~all
productions in $\tau_1,\ldots,\tau_k$ only use non-terminals from $M$. 
We say that a subset $Q\subseteq P$ is \emph{$M$-cancelable} if all productions in $Q$ 
are $M$-cancelable.

\subsubsection{The reduction.} 
Using the notions of $M$-cancelable productions, we are ready 
to describe how the context-free grammars are constructed. Suppose that $M$ is realizable,
that $Q\subseteq P$ is $M$-cancelable, and that $A\in M$. 
Consider the language
\[ 
	L_{A,Q} = \{u, v \in\Sigma^* \mid \exists~\text{derivation}~A\derivs[Q] uAv\}.
\]
Thus $L_{A,Q}$ consists of all words $u$ and $v$ appearing in derivations
(whose counter values are not necessarily zero) of the form $A\derivs uAv$, if
we only use $M$-cancelable productions. The $L_{A,Q}$ will be the languages
$L(G')$ mentioned above. 

It is an easy observation that we can, given $G$ and a subset $Q\subseteq P$,
construct a context-free grammar for $L_{A,Q}$:
\begin{lemma}
Given a $\Z$-grammar $G$, a non-terminal $A$, and a subset $Q\subseteq P$,
we can construct in logspace a context-free grammar for $L_{A,Q}$. Moreover, if
$G$ is left-linear, then the construction yields an NFA for $L_{A,Q}$.
\end{lemma}
We provide details in \cref{appendix-unboundedness-reduction}.
Now, our reduction works as follows:
\begin{enumerate}
\item Guess a subset $M\subseteq N$ and an $A\in M$; verify that $M$ is realizable.
\item Guess a subset $Q\subseteq P$; verify that $Q$ is $M$-cancelable.
\item Compute a context-free grammar for $L_{A,Q}$.
\end{enumerate}
Here, we need to show that steps 1 and 2 can be done in $\NP$:
\begin{restatable}{lemma}{realizableCancelable}\label{realizable-cancelable}
Given a subset $M\subseteq N$, we can check in $\NP$ whether $M$ is realizable.
Moreover, given $M\subseteq N$ and $p\in P$, we can check in $\NP$ if $p$ is $M$-cancelable.
\end{restatable}
Both can be done using the fact that for a given context-free grammar, one can construct a Parikh-equivalent existential Presburger formula~\cite{DBLP:conf/cade/VermaSS05} and the fact that satisfiability of existential Presburger formulas is in $\NP$. See \cref{appendix-unboundedness-reduction} for details.
This completes the description of our reduction.
Therefore, it remains to show correctness of the reduction. In other words, to prove:
\begin{proposition}\label{unboundedness-reduction-correctness}
We have $\pP(L(G))$ if and only if $\pP(L_{A,Q})$ for some subset $Q\subseteq
P$ such that there is a realizable $M\subseteq N$ with $A\in M$ and
$Q$ being $M$-cancelable.
\end{proposition}
\Cref{unboundedness-reduction-correctness} will be shown in two lemmas:
\begin{lemma}\label{unboundedness-factors}
If $M$ is realizable and $Q$ is $M$-cancelable, then $L_{A,Q}\subseteq F(L(G))$ for every $A\in M$.
\end{lemma}

\begin{lemma}\label{unboundedness-decomposition}
$L(G)$ is included in a finite union of sets of the form $K_1 \cdot K_2 \cdots
K_m$, where each $K_i$ is either finite or a set $L_{A,Q}$, where $Q$ is
$M$-cancelable for some realizable $M\subseteq N$, and $A\in M$.
\end{lemma}

Let us see why \cref{unboundedness-reduction-correctness} follows from
\cref{unboundedness-factors,unboundedness-decomposition}.
\begin{proof}[\cref{unboundedness-reduction-correctness}]
We begin with the ``if'' direction. Thus, suppose $\pP(L_{A,Q})$ for $A$ and
$Q$ as described.  Then by \cref{unboundedness-factors} and the first and fourth axioms of
unboundedness predicates, this implies $\pP(L(G))$. 

For the ``only if'' direction, suppose $\pP(L(G))$. By the first axiom of unboundedness predicates, $\pP$ must hold for the finite union provided by \cref{unboundedness-decomposition}. By the second axiom, this implies that $\pP(K_1\cdots K_m)$ for a finite product
$K_1\cdots K_m$ as in \cref{unboundedness-decomposition}. Moreover, by
the third axiom, this implies that $\pP(K_i)$ for some $i\in\{1,\ldots,m\}$.
If $K_i$ is finite, then by assumption, $\pP(K_i)$ does not hold. Therefore, we
must have $\pP(K_i)$ for some $K_i=L_{A,Q}$, as required.
\end{proof}

\subsubsection{Flows.}
It remains to prove \cref{unboundedness-factors,unboundedness-decomposition}.
We begin with \cref{unboundedness-factors} and for this we need some more terminology.
Let $\Sigma$ be an alphabet. By $\Parikh\colon\Sigma^*\to\N^\Sigma$, we denote
the \emph{Parikh map}, which is defined as $\Parikh(w)(a)=|w|_a$ for
$w\in\Sigma^*$ and $a\in\Sigma$. In other words, $\Parikh(w)(a)$ is the number
of occurrences of $a$ in $w\in\Sigma^*$. If $\Gamma\subseteq\Sigma$ is a subset, then $\pi_\Gamma\colon\Sigma^*\to\Gamma^*$ is the  homomorphism with $\pi_\Gamma(a)=\varepsilon$ for $a\in\Sigma\setminus\Gamma$
and $\pi_\Gamma(a)=a$ for $a\in\Gamma$. We also call $\pi_\Gamma$ the \emph{projection to $\Gamma$}.

Suppose we have a $\Z$-grammar $G=(N,\Sigma,P,S)$ with non-terminals $N$ 
and productions $P$.  For a derivation tree $\tau$, we write $\Parikh(\tau)$ for the vector 
in $\N^P$ that counts how many times each production appears in $\tau$.
We introduce a map $\partial$, which counts how many non-terminals each production consumes and produces. Formally, $\partial\colon\N^P\to\Z^N$ is the monoid homomorphism that sends the production $p=A\to w$
to the vector $\partial(p)=-A+\Parikh(\pi_N(w))$. Here, $-A\in\Z^N$ denotes the vector with $-1$ at the position of $A$ and $0$ everywhere else.
A vector $\bu\in\N^P$ is a \emph{flow} if $\partial(\bu)=\bzero$.
Observe that a derivation tree $\tau$ is a pump if and only if $\Parikh(\tau)$ is a flow.
In this case, we also call the vector $\bu\in\N^P$ with $\bu=\Parikh(\tau)$ a \emph{pump}.

The following \lcnamecref{euler} will provide an easy way to construct
derivations. It is a well-known result by
Esparza~\cite[Theorem 3.1]{esparza1997petri}, and has since been exploited in
several results on context-free grammars. Our formulation is slightly
weaker than Esparza's. However,
it is enough for our purposes and admits a simple proof, which is
inspired by a proof of Kufleitner~\cite{KufleitnerEuler}.
\begin{lemma}\label{euler}
	Let $\bmf\in\N^P$. Then $\bmf$ is a flow if and only if it is a sum of pumps.
\end{lemma}
\begin{proof}
The ``if'' direction is trivial, because every pump is clearly a flow.
Conversely, suppose $\bmf\in\N^P$ is a flow. We can clearly write
$\bmf=\Parikh(\tau_1)+\cdots+\Parikh(\tau_n)$, where $\tau_1,\ldots,\tau_n$ are derivation trees:
We can just view each production in $\bmf$ as its own derivation tree. Now
suppose that we have $\bmf=\Parikh(\tau_1)+\cdots+\Parikh(\tau_n)$ so that $n$ is
minimal. We claim that then, each $\tau_i$ is a pump, proving the
\lcnamecref{euler}.

Suppose not, then without loss of generality, $\tau_1$ is not a pump. Since $\tau_1$
is a derivation, this means $\Parikh(\tau_1)$ cannot be a flow and thus there must
be a non-terminal $A$ with $\partial(\tau_1)(A)\ne 0$. 

Let us first assume that $\partial(\tau_1)(A)>0$.  This means there is a
non-terminal $A$ occurring at a leaf of $\tau_1$ such that $A$ is not the start
symbol of $\tau_1$. Since $\bmf=\Parikh(\tau_1)+\cdots+\Parikh(\tau_n)$ is a flow, we
must have $\partial(\Parikh(\tau_2)+\cdots+\Parikh(\tau_n))(A)<0$. This, in turn, is
only possible if some $\tau_j$ has $A$ as its start symbol. We can therefore merge
$\tau_1$ and $\tau_j$ by replacing $\tau_1$'s $A$-labelled leaf by the new subtree
$\tau_j$. We obtain a new collection of $n-1$ trees whose Parikh image is $\bmf$,
in contradiction to the choice of $n$.  If $\partial(\tau_1)(A)<0$, then there
must be a $\tau_j$ with $\partial(\tau_j)(A)>0$ and thus we can insert $\tau_1$
below $\tau_j$, reaching a similar contradiction.
\end{proof}

\subsubsection{Constructing derivations.}
Using flows, we can now prove \cref{unboundedness-factors}. 
\begin{proof}
Suppose there is a derivation $\tau\colon A\derivs[Q] uAv$ with $A\in M$ and $u,v\in\Sigma^*$. We have to show
that both $u$ and $v$ occur in some word $w\in L(G)$.  Furthermore, if $G$ is
in Chomsky normal form, we can choose $w$ such that $|w|$ is linear in $|u|$ and
$|v|$.  Our goal is to construct a derivation of $G$ in which we find $u$ and
$v$ as factors.  We could obtain a derivation tree by inserting $\tau$ into
some derivation tree for $G$ (at some occurrence of $A$), but this might yield
non-zero counter values. Therefore, we will use the fact that $Q$ is
$M$-cancelable to find other pumps that can be inserted as well in order to
bring the counter back to zero.

Since $M\subseteq N$ is realizable, there exists a complete derivation $\tau_0$
that derives some word $w_0\in L(G)$ and uses precisely the non-terminals in
$M$. Since $Q\subseteq P$ is $M$-cancelable, we
know that for each production $p\in Q$, there exist pumps $\tau_1,\ldots,\tau_k$
such that (i)~$p$ occurs in some $\tau_i$,
(ii)~$\counters(\tau_1)+\cdots+\counters(\tau_k)=\bzero$ and (iii)~all
productions in $\tau_1,\ldots,\tau_k$ only use non-terminals in $M$. This
allows us to define $\bmf_p:=\Parikh(\tau_1)+\cdots+\Parikh(\tau_k)$. Observe that 
$\vec{f}_p$ contains only productions with non-terminals from $M$, we have $\vec{f}_p[p] > 0$, and $\counters(\vec{f}_p) = 0$.
We can use the flows $\bmf_p$ to find the desired canceling pumps.
Since by \cref{euler}, every flow can be decomposed
into a sum of pumps, it suffices to construct a particular flow. Specifically,
we look for a flow $\bmf_\tau \in \N^P$ such that:
	\begin{enumerate}
		\item any production $p$ with $\bmf_\tau[p] > 0$ uses only non-terminals from $M$, and
		\item $\counters(\bmf_\tau + \Parikh(\tau)) = 0$.
	\end{enumerate}
The first condition ensures that all the resulting pumps can be inserted into $\tau_0$. The second condition ensures that the resulting total counter values will be zero. We claim that with
	\begin{equation}
		\label{eq:ftau}
		\bmf_\tau = \left(\sum_{p \in Q} \Parikh(\tau)[p] \cdot \vec{f}_p\right) - \Parikh(\tau),
	\end{equation}
we achieve these conditions.
First, observe that $\bmf_\tau\in\N^P$: We have
	\[ \begin{aligned}
		\bmf_\tau[q] ~~\geq~~ \Parikh(\tau)[q] \cdot \vec{f}_q[q] - \Parikh(\tau)[q] 
		~~=~~\Parikh(\tau)[q] \cdot (\vec{f}_q[q] - 1)
	\end{aligned} \]
	which is at least zero as $\vec{f}_q[q]$ must be non-zero by definition.
	Second, note that $\bmf_\tau$ is indeed a flow, because it is a $\Z$-linear combination of flows.
Moreover, all productions appearing in $\bmf_\tau$ also appear in $\bmf_p$ for some $p\in Q$ or in $\tau$, meaning that all non-terminals must belong to $M$. Finally, the total counter effect of $\bmf_\tau + \Parikh(\tau)$ is zero as $\bmf_\tau + \Parikh(\tau) = \sum_{p \in Q} \Parikh(\tau)[p] \cdot \vec{f}_p$ is a sum of flows each with total counter effect zero.

	Now, since $\bmf_\tau$ is a flow, \cref{euler} tells us that there are pumps $\tau'_1,\ldots,\tau'_m$ such that $\bmf_\tau=\Parikh(\tau'_1)+\cdots+\Parikh(\tau'_m)$. Therefore, inserting $\tau$ and $\tau'_1,\ldots,\tau'_m$ into $\tau_0$ must yield a derivation of a word that has both $u$ and $v$ as factors
and also has counter value 
\[ \underbrace{\counters(\tau_0)}_{=\bzero}+\underbrace{\counters(\tau)+\counters(\tau'_1)+\cdots\counters(\tau'_m)}_{=\counters(\tau)+\counters(\bmf_\tau)=\bzero}=\bzero.\]
Thus, we have a complete derivation of $G$.
Hence $L_{A,Q}\subseteq F(L(G))$.
\end{proof}

\subsubsection{Decomposition into finite union.}
It remains to prove \cref{unboundedness-decomposition}. For the decomposition,
we show that there exists a finite set $D_0$ of complete derivations such that
all complete derivations of $G$ can be obtained from some derivation in $D_0$
and then inserting pumps that produce words in $L_{A,Q}$, for some appropriate
$A$ and $Q$. Here, it is key that the set $D_0$ of ``base derivations'' is finite.
Showing this for context-free grammars would just require a simple
``unpumping'' argument based on the pigeonhole principle as in Parikh's
theorem~\cite{parikh1966context}.  However, in the case of $\Z$-grammars, where
$D_0$ should only contain derivations that have counter value zero, this is
not obvious. To achieve this, we employ a well-quasi ordering on (labeled)
trees.  Recall that a \emph{quasi ordering} is a reflexive and transitive
ordering.  For a quasi ordering $(X,\le)$ and a subset $Y\subseteq X$, we write
$\upcl{Y}$ for the set $\{x\in X\mid \exists y\in Y\colon y\le x\}$.  We say
that $(X,\le)$ is a \emph{well-quasi ordering} (WQO) if every non-empty subset
$Y\subseteq X$ has a finite subset $Y_0\subseteq Y$ such that $Y\subseteq
\upcl{Y_0}$.

We define an ordering on all trees in $\cT_{N\cup \Sigma,d}$. A tree \emph{s} is a \emph{subtree}
of $t$ if there exists a node $x$ in $t$ such that $s$ consists of all nodes of
$t$ that are descendants of $x$. If $\tau_1,\ldots,\tau_n$ are trees, then we denote
by $r[\tau_1,\ldots,\tau_n]$ the tree with a root node $r$ and the subtrees $\tau_1,\ldots,\tau_n$ directly under the root.
Now let $\tau=(A,\bu)[\tau_1,\ldots,\tau_n]$ and $\tau'=(B,\bv)[\sigma_1,\ldots,\sigma_m]$ be trees in $\cT_{N\cup \Sigma,d}$. We define the ordering $\preceq$ as follows. If $n=0$ (i.e.\ $\tau$ consists of only one node), then we have $\tau\preceq\tau'$ if and only if $A=B$ and $m=0$. If $n\ge 1$, then we define inductively: 
\[\begin{aligned}
	\tau \preceq \tau' ~~~\iff~~~ A = B~\text{ and }\exists\text{ subtree}~&\tau'' = (A,\bu')[\tau'_1, \ldots, \tau'_n]~\text{of}~\tau'\\ 
&                           \text{with}~\tau_i \preceq \tau'_i~\text{for $i=1,\ldots,n$}
\end{aligned}\]
Based on $\preceq$, we define as slight refinement: We write
$\tau\pumpleq\tau'$ if and only if $\tau\preceq\tau'$ and the set of
non-terminals appearing in $\tau$ is the same as in $\tau'$.
\begin{lemma}\label{wqo}
	$(\cT_{N\cup\Sigma,d},\pumpleq)$ is a WQO.
\end{lemma}
\begin{proof}
	In \cite[Lemma 3.3]{DBLP:journals/corr/abs-1904-04090}, it was shown that $\preceq$ is a WQO. Then $\pumpleq$
is the product of equality on a finite set, which is a WQO, and
the WQO $\preceq$.
\end{proof}

\cref{wqo} allows us to decompose $L(G)$ into a finite union: 
For each complete derivation $\tau$ of $G$, we define
\[ L_\tau(G) =\{w\in\Sigma^* \mid \exists~\text{complete derivation $\tau'$ with $\tau\pumpleq\tau'$ and $\yield(\tau')=w$}\}. \]

\begin{lemma}\label{unboundedness-decomposition-union}
There exists a finite set $D_0\subseteq\cT_{N\cup\Sigma,d}$ of complete derivations of $G$ such that
$L(G)=\bigcup_{\tau\in D_0} L_\tau(G)$.
\end{lemma}
\begin{proof}
Since $(\cT_{N\cup\Sigma,d},\pumpleq)$ is a WQO, the set $D\subseteq\cT_{N\cup T,d}$ of all complete derivations of $G$
has a finite subset $D_0$ with $D\subseteq\upcl{D_0}$. This implies
the \lcnamecref{unboundedness-decomposition-union}.
\end{proof}

\subsubsection{Decomposition into finite product.}
In light of \cref{unboundedness-decomposition-union}, it remains to be shown
that for each tree $\tau$, we can find a product $K_1 \cdot K_2 \cdots K_m$ of languages
such that $L_\tau(G)\subseteq K_1 \cdot K_2 \cdots K_m$ and each $K_i$ is either finite or
is of the form $L_{A,Q}$. We construct the overapproximation of $L_\tau(G)$
inductively as follows. Let $M\subseteq N$ and $Q\subseteq P$ be subsets of the non-terminals and
the productions, respectively. If $\tau$ has one node, labeled by $a\in\Sigma$,
then we set $\App_{Q}(\tau):=\{a\}$. Moreover, if $\tau=(A,\bu)[\tau_1,\ldots,\tau_n]$ for $A\in N$ and trees $\tau_1,\ldots,\tau_n$, then we set
\[ \App_{Q}(\tau):=L_{A,Q} \cdot \App_{Q}(\tau_1) \cdot \App_{Q}(\tau_2) \cdots\App_Q(\tau_n)\cdot L_{A,Q}. \]
Finally, we set $\App(\tau):=\App_Q(\tau)$, where $Q\subseteq P$ is the set of
all $M$-cancelable productions, where $M$ is the set of all non-terminals
appearing in $\tau$.  Now clearly, each $\App(\tau)$ is a finite product
$K_1 \cdot K_2 \cdots K_m$ as desired: This follows by induction on the size of $\tau$.
Thus, to prove \cref{unboundedness-decomposition}, the following suffices:
\begin{lemma}
For every complete derivation tree $\tau$ of $G$, we have $L_\tau(G)\subseteq \App(\tau)$.
\end{lemma}
\begin{proof}
Suppose $w\in L_\tau(G)$ is derived using a complete derivation tree $\tau'$
with $\tau\pumpleq\tau'$.  Then, the set of non-terminals appearing in
$\tau$ must be the same as in $\tau'$; we denote it by $M$. Let $Q\subseteq P$
be the set of all $M$-cancelable productions.  Moreover, since
$\tau\preceq\tau'$, we can observe that there exist pumps
$\tau_1,\ldots,\tau_n$ with root non-terminals $A_1,\ldots,A_n$ and nodes
$x_1,\ldots,x_n$ in $\tau$ such that $\tau'$ can be obtained from $\tau$ by
replacing each node $x_i$ by the pump $\tau_i$.

Since both $\tau$ and $\tau'$ are complete derivations of $G$, each must have
counter effect $\bzero$.  Thus,
$\counters(\tau_1)+\cdots+\counters(\tau_n)=\counters(\tau')-\counters(\tau)=\bzero$.
Hence, the pumps $\tau_1,\ldots,\tau_n$ witness that the productions
appearing in $\tau_1,\ldots,\tau_n$ are $M$-cancelable. Thus, the derivation
corresponding to $\tau_i$ uses only productions in $Q$ and thus $\tau_i$ corresponds to
$A_i\derivs[Q]u_iAv_i$ for some $u_i,v_i$ and we have $u_i,v_i\in L_{A,Q}$.
\end{proof}

%% file: growth.tex
In this section, we prove \cref{main-growth}.  Since clearly, a bounded
language has polynomial growth, it remains to be shown that if $L$ is accepted
by a PRBCA and $L$ is not bounded, then it has exponential growth.  
For two languages $L_1,L_2\subseteq\Sigma^*$, we write $L_1\linembed L_2$ if
there exists a constant $c\in\N$ such that for every word $w_1 \in L_1$, there
exists $w_2 \in L_2$ with $|w_2| \le c\cdot |w_1|$ and $w_1$ is a factor of
$w_2$.  It is not difficult to observe that for two languages
$L_1,L_2\subseteq\Sigma^*$, if $L_1\linembed L_2$  and $L_1$ has exponential
growth, then so does $L_2$. 

In order to show \cref{main-growth}, we need an adapted version of
\cref{unboundedness-factors}.  A $\Z$-grammar is in \emph{Chomsky normal form}
if all productions are of the form $(A\to BC,\bv)$ or $(A\to a,\bv)$ with
$A,B,C\in N$, $a\in\Sigma$, and $\bu,\bv\in\Z^k$. In other words, the
context-free grammar obtained by forgetting all counter vectors is in Chomsky
normal form. Fernau and
Stiebe~\cite[Proposition 5.12]{DBLP:journals/tcs/FernauS02} have shown that every
$\Z$-grammar has an equivalent $\Z$-grammar in Chomsky normal form.
\begin{lemma}\label{unboundedness-factors-stronger}
If $G=(N,\Sigma,P,S)$ is a $\Z$-grammar in Chomsky normal form, 
$M\subseteq N$ is realizable, $Q\subseteq P$ is $M$-cancelable, and $A\in M$,
then $L_{A,Q}\linembed L(G)$.
\end{lemma}
This is shown essentially the same way as \cref{unboundedness-factors}.  
Let us now show that if a language $L$ accepted by a PRBCA is not bounded, then
it must have exponential growth. We have seen above that as a PRBCA language,
$L$ is generated by some $\Z$-grammar. As shown by Fernau and
Stiebe~\cite[Proposition 5.12]{DBLP:journals/tcs/FernauS02}, this implies that
$L=L(G)$ for some $\Z$-grammar $G$ in Chomsky normal form. Since $L$ is not
bounded, \cref{unboundedness-decomposition} yields $A$ and $Q$ such that
$L_{A,Q}$ is not a bounded language. It is well-known that any context-free
languages that is not bounded has exponential growth (this fact has apparently
been independently discovered at least six times,
see~\cite{DBLP:journals/ijfcs/GawrychowskiKRS10} for references). Thus,
$L_{A,Q}$ has exponential growth.
By \cref{unboundedness-factors-stronger}, we have $L_{A,Q}\linembed L$ and
thus $L$ has exponential growth.

%% file: appendix-results.tex
\subsection{Counter updates} \label{appendix-counter-updates}

Here we argue that for $d$-dimensional $\Z$-VASS we can without loss of generality
assume unary encoded counter updates.
In particular, we show that counter updates comprised of $n$-bit binary encoded numbers
can be simulated by ones comprised of unary encoded numbers.
Moreover, we also mention how to extend this result to the case of $\Z$-grammars, which
uses a very similar construction.
The general idea is to replace each counter by $n$ many counters and to replace each
number in a counter update by its binary representation.

Let $\cV=(Q,\Sigma,q_0,T,F)$ be a $d$-dimensional $\Z$-VASS with binary encoded counter
updates using at most $n$ bit per number.
We construct a $dn$-dimensional $\Z$-VASS $\cV'=(Q,\Sigma,q_0,T',F)$ with unary
encoded counter updates as follows:
\begin{itemize}
  \item For every transition $(p,w,(v_1,\ldots,v_d),p') \in T$ add the transition
  \linebreak $(p,w,(\bin(v_1), \ldots, \bin(v_d)),p')$ to $T'$, where $\bin:\Z \to \{-1,0,1\}$ maps an $n$-bit number to its binary representation, while keeping the sign.
  \item For every $i \in [1,d]$, every $j \in [1,n-1]$ and every state $p \in Q$
  add the loops $(p,\varepsilon,(u_1,\ldots,u_{dn}),p)$ and
  $(p,\varepsilon,(v_1,\ldots,v_{dn}),p)$ to $T'$,
  where $u_{(i-1)n + j} = -2$, $u_{(i-1)n + j + 1} = 1$, $v_{(i-1)n + j} = 2$, $v_{(i-1)n + j + 1} = -1$,
  and $u_k = v_k = 0$ for every $k \notin \{(i-1)n + j, (i-1)n + j + 1\}$.
\end{itemize}
Here the first set of transitions applies the binary representations as counter effects,
while the loops of the second set ensure the correct relationship among each set of $n$
counters simulating a single counter of $\cV$.
In particular, the loops make sure that each such counter represents one of the $n$ bits,
with each successive bit being twice as significant as the previous one.

Let us now briefly talk about the construction for $\Z$-grammars.
The idea is the same, but instead of replacing transitions, we replace productions.
Then instead of the loops, we add for each non-terminal $A$ productions that replace $A$ by itself.
In both cases the counter updates are exactly the same as in the $\Z$-VASS construction,
and the proof is very similar as well.

To prove the construction correct in the $\Z$-VASS case we show
by induction on run length that there is a run
$(q,u_1,\ldots,u_n) \xrightarrow{w}(p,v_1,\ldots,v_n)$ of $\cV$
if and only if there is a run
$(q,u_1',\ldots,u_{dn}') \xrightarrow{w}(p,v_1',\ldots,v_{dn}')$ of $\cV'$
with $u_i = \sum_{j = 1}^{n} 2^{j-1} u_{(i-1)n + j}'$ and
$v_i = \sum_{j = 1}^{n} 2^{j-1} v_{(i-1)n + j}'$.
The base case is immediate from the starting configuration in both $\Z$-VASS.

For the inductive case let us firstly assume that there is a run
$(q,u_1,\ldots,u_n)$ $\xrightarrow{w}(p,v_1,\ldots,v_n)
\xrightarrow {w'} (r,x_1,\ldots,x_n)$ of $\cV$
such that the second part is due to a single transition
$(p,w',\vec x - \vec v,r) \in T$.
Then by induction hypothesis there is a run
$(q,u_1',\ldots,u_{dn}') \xrightarrow{w}(p,v_1',\ldots,v_{dn}')$
of $\cV'$ with $u_i = \sum_{j = 1}^{n} 2^{j-1} u_{(i-1)n + j}'$ and
$v_i = \sum_{j = 1}^{n} 2^{j-1} v_{(i-1)n + j}'$.
Furthermore we have
$(p,w',(\bin(x_1 - v_1), \ldots, \bin(x_d - v_d)),r) \in T'$ by construction.
Appending this transition to the run of $\cV'$ ending in $(p,\vec v')$
leads to the configuration $(r,y_1',\ldots,y_{dn}')$, where
$y_{(i-1)n + j}' = v_{(i-1)n + j}' + \bin(x_i - v_i)_j$
for $i \in [1,d]$ and $j \in [1,n]$.
Here $\bin(-)_j$ refers to the $j$th digit in the binary representation,
which corresponds to the $(j-1)$th power of $2$.
This gives us the following relationship between $\vec x$ and $\vec y'$,
as required:
\begin{align*}
  \sum_{j = 1}^{n} 2^{j-1} y_{(i-1)n + j}' &= \sum_{j = 1}^{n} 2^{j-1} v_{(i-1)n + j}' + 2^{j-1} \bin(x_i - v_i)_j\\
  &= (x_i - v_i) + \sum_{j = 1}^{n} 2^{j-1} v_{(i-1)n + j}' = (x_i - v_i) + v_i = x_i.
\end{align*}

Secondly, let us now assume that there is a run
$(q,u_1',\ldots,u_{dn}') \xrightarrow{w}$ \linebreak $(p,v_1',\ldots,v_{dn}')
\xrightarrow {w'} (r,y_1',\ldots,y_{dn}')$ of $\cV'$
such that the second part is due to a single transition
$t' = (p,w',\vec y' - \vec v',r) \in T'$.
Then by induction hypothesis there is a run
$(q,u_1,\ldots,u_n) \xrightarrow{w}(p,v_1,\ldots,v_n)$ of $\cV$
with $u_i = \sum_{j = 1}^{n} 2^{j-1} u_{(i-1)n + j}'$ and
$v_i = \sum_{j = 1}^{n} 2^{j-1} v_{(i-1)n + j}'$.
Now we need to distinguish two cases:
\begin{itemize}
  \item If $t'$ was added to $T'$ based on a transition $t \in T$, then we have
  $t = (p,w',\Big(\sum_{j = 1}^{n} 2^{j-1}(y_j' - v_j'), \ldots,
  \sum_{j = 1}^{n} 2^{j-1}(y_{(d-1)n + j}' - v_{(d-1)n + j}')\Big),r)$ by construction.
  Appending this transition to the run of $\cV$ ending in $(p,\vec v)$
  leads to the configuration $(r,x_1,\ldots,x_d)$, where
  $x_i = v_i + \sum_{j = 1}^{n} 2^{j-1}(y_{(i-1)n + j}' - v_{(i-1)n + j}')$
  for $i \in [1,d]$.
  This gives us the following relationship between $\vec x$ and $\vec y'$,
  as required:
  \begin{align*}
    x_i &= v_i + \sum_{j = 1}^{n} 2^{j-1}(y_{(i-1)n + j}' - v_{(i-1)n + j}')\\
    &= (v_i - v_i) + \sum_{j = 1}^{n} 2^{j-1}y_{(i-1)n + j}'
    = \sum_{j = 1}^{n} 2^{j-1}y_{(i-1)n + j}'.
  \end{align*}
  
  \item If $t'$ is one of the loops in $T'$ that was no based on a transition of $T$,
  then we have $w' = \varepsilon$, $p = r$, and the counter effect $\vec y' - \vec v'$
  changes the value $v_{(i-1)n + j}'$ by $2$, and the value $v_{(i-1)n + j + 1}'$ by $1$,
  for some $i \in [1,d]$, $j \in [1,n-1]$.
  Since all other values are not changed by this effect, it only affects the sum
  $\sum_{j = 1}^{n} 2^{j-1} v_{(i-1)n + j}'$. Since the change by $2$ and the change by
  $1$ always have opposite signs, this sum is also not affected.
  Therefore the run of $\cV$ ending in $(p,\vec v)$ is already as required.
\end{itemize}
For language equivalence between $\cV$ and $\cV'$, we still need to show that if $\cV$
reaches $\vec 0$ then $\cV'$ can simulate this by reaching exactly $\vec 0$ as well.
Formally we show that there is a run $(q_0,\vec 0) \xrightarrow{w}(p,\vec 0)$ of $\cV$
if and only if there is a run $(q_0,\vec 0) \xrightarrow{w}(p,\vec 0)$ of $\cV'$.

The ``only if''-direction follows directly from the statement we just proved.
For the ``if''-direction we get that there is a run
$(q_0,\vec 0) \xrightarrow{w}(p,\vec v')$ of $\cV'$
with $\sum_{j = 1}^{n} 2^{j-1} v_{(i-1)n + j}' = 0$ for every $i \in [1,d]$.
Assume that for some $i \in [1,d]$ there is a maximal $j \in [1,n]$ such that
$v_{(i-1)n + j}' \neq 0$.
Then $j > 1$, because the above sum must evaluate to $0$, so there must be a
smaller $j' \in [1,n]$ with $v_{(i-1)n + j'}' \neq 0$ as well.
Consider the loop
$\ell = (p,\varepsilon,(u_1,\ldots,u_{dn}),p) \in T'$
where $u_{(i-1)n + j - 1} = -2$, $u_{(i-1)n + j} = 1$ if $v_{(i-1)n + j} < 0$,
or $u_{(i-1)n + j - 1} = 2$, $u_{(i-1)n + j} = -1$ if $v_{(i-1)n + j} < 0$.
Appending $\ell$ $|v_{(i-1)n + j}|$ many times to the run ending in $(p,\vec v')$
will change the value $v_{(i-1)n + j}$ to $0$, and otherwise only changes
values $v_{(i-1)n + j''}'$ with $j'' \in [1,n]$ and $j'' < j$.
However, this run extension does not alter the sum
$\sum_{j = 1}^{n} 2^{j-1} v_{(i-1)n + j}'$. Meaning if we do this iteratively for
every $j$ and $i$, we eventually arrive at $\vec 0$.

\subsection{Unboundedness predicates} \label{appendix-unboundedness-predicates}

We comment here on the relationship between the notion of
one-dimensional unboundedness predicates in~\cite{DBLP:conf/icalp/CzerwinskiHZ18}
and the notion used here. The set of axioms in \cite{DBLP:conf/icalp/CzerwinskiHZ18}
is as follows, for any languages $K,L\subseteq\Sigma^*$:
\begin{enumerate}[leftmargin=*, label=(O\arabic*)]
\item If $\pP(K)$ and $K\subseteq L$, then $\pP(L)$.
\item If $\pP(K\cup L)$, then $\pP(K)$ or $\pP(L)$.
\item If $\pP(F(K\cdot L))$, then $\pP(F(K))$ or $\pP(F(L))$.
\end{enumerate}
Moreover, in \cite{DBLP:conf/icalp/CzerwinskiHZ18}, the resulting decision problem is defined as:
\begin{description}
\item[Input] Language $L\subseteq\Sigma^*$.
\item[Question] Does $\pP(F(L))$ hold?
\end{description}
Note that in this definition, one decides whether $\pP(F(L))$, in contrast to
the definition in \cref{sec:results}, where we decide whether $\pP(L)$.
\begin{enumerate}
\item The definition (U1)--(U4) is slightly stronger. Observe that a predicate
$\pP$ that satisfies (U1)--(U4) also satisfies (O1)--(O3). However, there are some (pathological) examples of predicates that satisfy (O1)--(O3), but not (U1)--(U4). For example, the predicate $\pP_{\neg a}$ with
\[ \pP_{\neg a}(L) \iff \text{$L$ is non-empty and $L\ne\{a\}$} \]
satisfies (O1)--(O3), but not (U4). The latter is because $F(\{a\})=\{\varepsilon,a\}$
satisfies $\pP_{\neg a}$, but $\{a\}$ itself does not.
\item Although (U1)--(U4) are slightly stronger, they yield the same set of
decision problems. Indeed, suppose $\pP$ is a predicate that satisfies
(O1)--(O3). Define the predicate $\pP'$ with $\pP'(L)\iff \pP(F(L))$. Then,
$\pP'$ satisfies (U1)--(U4). Moreover, the decision problem induced by $\pP$
and $\pP'$ is the same: In each case, we decide whether $\pP(F(L))$.
\item Thus, the set of decision problems induced by the axiom sets (O1)--(O3)
and (U1)--(U4) is the same, and hence \cref{main-unboundedness} is independent
of which set of axioms is used.
\end{enumerate}

The reason the axioms (O1)--(O3) were used in
\cite{DBLP:conf/icalp/CzerwinskiHZ18} is that if one also introduces
multi-dimensional unboundedness predicates, the axioms (O1)--(O3) fit more
nicely with the general case. The slight disadvantage of the definition in
\cite{DBLP:conf/icalp/CzerwinskiHZ18} is that one has to decide $\pP(F(L))$
instead of $\pP(L)$, which is why here, we opted for (U1)--(U4).

%% file: app-rbc-reversal-reduction.tex
\subsubsection{Translating transitions of $(k, r)$-RBCA to $(k', 1)$-RBCA.}

We provide an explicit construction for the four transition types $\mathsf{inc}_i$, $\mathsf{dec}_i$, $\mathsf{zero}_i$, and $\mathsf{nz}_i$.

Let us first consider a transition $(p, a, \mathsf{inc}_i, q)$ (\cref{fig:inc-gadget}). We start with a sequence of tests (marked in blue) for $\pctr_{i,j}$ to identify which phase we are currently in. If $\pctr_{i,j}$ is non-zero for an odd $j$, we know we are in an positive phase and simply increment the corresponding program counter $\ctr_{i,\frac{j+1}{2}}$. Otherwise we are in a negative phase. Thus we need to switch to an positive phase, decrementing $\pctr_{i,j}$ and incrementing $\pctr_{i,j+1}$, before incrementing the corresponding program counter $\ctr_{i,\frac{j}{2}+1}$.

Observe that the state $f$ can never be reached in an $r$-reversal bounded run. As the phase counters $\pctr_{i,1}$ to $\pctr_{i,2m-1}$ have all tested negative, we must be in phase $r = 2m$. As this is an even number, it must be descending, and incrementing would mean adding another reversal, breaking the bound.

\begin{figure}
	\begin{center}
			\begin{tikzpicture}[state/.style={draw,shape=circle,minimum width=1em},auto, node distance=5em]
			\node [state] (p) {$p$};
			\node [state, right=2em of p, fill=blue!10!white] (r1) {};
			\node [state, below=10em of r1] (s1) {};
			\node [state, right=of r1, fill=blue!10!white] (r2) {};
			\node [state, below=5em of r2] (s2) {};
			\node [state, right=of s2] (s3) {};
			\node [state, right=of s3] (s4) {};
			\node [right=of r2, fill=blue!10!white] (r4) {$\ldots$};
			\node [state, right=of r4,fill=red!10!white] (r5) {$f$};
			\node [state, right=6em of r5] (q) {$q$};
			
			\draw [->] (p) edge node {$a$} (r1);
			\draw [->] (r1) edge node [yshift=12,xshift=-5,rotate=-45] {\texttt{$\pctr_{i,1}$ = 0?}} (r2);
			\draw [->] (r2) edge node [yshift=12,xshift=-5,rotate=-45] {\texttt{$\pctr_{i,2}$ = 0?}} (r4);
			\draw [->] (r4) edge node [yshift=16,xshift=-10,rotate=-45] {\texttt{$\pctr_{i,2m-1}$ = 0?}} (r5);
			
			\draw [->] (r1) edge node [left,xshift=-3,yshift=-15,rotate=-45,] {\texttt{$\pctr_{i,1}$} $\neq$ 0?} (s1);
			\draw [->] (s1)  -| node [near start]{\texttt{$\ctr_{i,1}$ += 1}} (q);
			
			\draw [->] (r2) edge node [left,xshift=-3,yshift=-15,rotate=-45,] {\texttt{$\pctr_{i,2}$} $\neq$ 0?} (s2);
			\draw [->] (s2) edge node [yshift=16,xshift=-10,rotate=-45] {\texttt{$\pctr_{i,2}$ -= 1}} (s3);
			\draw [->] (s3) edge node [yshift=16,xshift=-10,rotate=-45] {\texttt{$\pctr_{i,3}$ += 1}} (s4);
			\draw [->] (s4) -| node [near start,yshift=16,xshift=-10,rotate=-45] {\texttt{$\ctr_{i,2}$ += 1}} (q);
		\end{tikzpicture}
		\caption{Transition structure to replace an incrementing transition $(p, a, \mathsf{inc}_i, q)$.}
		\label{fig:inc-gadget}
	\end{center}
\end{figure}
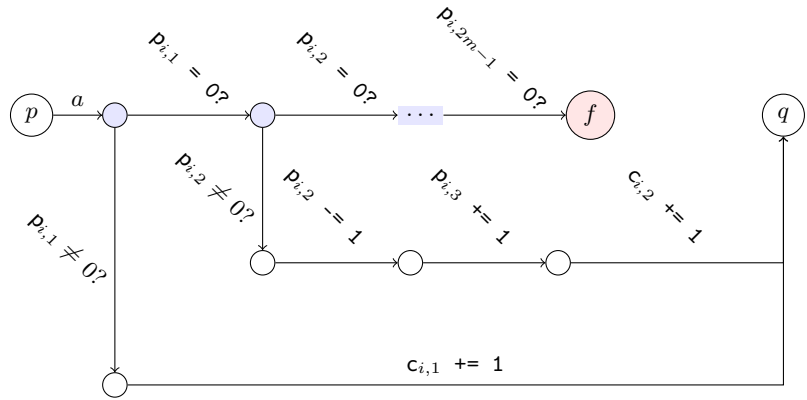

Next, let us consider a transition $(p, a, \mathsf{dec}_i, q)$ (\cref{fig:dec-gadget}). Similarly to the previous construction, we again first identify the current phase in the state sequence marked in blue. This time, we change the phase counter if we are currently in a positive (odd) phase and move to a negative phase. Once the phase counter is updated, we change to the states marked in green, which allow us to non-deterministically decrement one of the program counters (recall that if we are in phase $j$, we choose one of the counters $\ctr_{i,1}, \ldots, \ctr_{i,\lceil\frac{j}{2}\rceil}$ to decrement).

\begin{figure}
	\begin{center}
		\begin{tikzpicture}[state/.style={draw,shape=circle,minimum width=1em},auto, node distance=5em]
			\node [state] (p) {$p$};
			\node [state, right=2em of p, fill=blue!10!white] (r1) {};
			\node [state, below=of r1] (s1) {};
			\node [state, below=of s1] (s2) {};
			\node [state, below=of s2, fill=green!10!white] (s3) {};
			\node [state, right=of r1, fill=blue!10!white] (r2) {};
			\node [state, below=5em of r2] (s4) {};
			\node [right=of r2, fill=blue!10!white] (r4) {$\ldots$};
			\node [state, right=of r4, fill=blue!10!white] (r5) {};
			\node [state, below=of r5] (r51) {};
			\node [state, below=of r51] (r52) {};
			\node [state, below=of r52, fill=green!10!white] (dm) {};
			\node [state, right=of r5, fill=blue!10!white] (r6) {};  
			\node [state, below=6em of s3] (q) {$q$};
			\node [state,right=of s3, fill=green!10!white] (d2) {};
			\node [state,left=of dm, fill=green!10!white] (dm1) {};
			\node [left=2.5em of dm1, fill=green!10!white]  (t2) {$\ldots$};
			
			\draw [->] (p) edge node {$a$} (r1);
			\draw [->] (r1) edge node [yshift=12,xshift=-5,rotate=-45] {\texttt{$\pctr_{i,1}$ = 0?}} (r2);
			\draw [->] (r2) edge node [yshift=12,xshift=-5,rotate=-45] {\texttt{$\pctr_{i,2}$ = 0?}} (r4);
			\draw [->] (r4) edge node [yshift=16,xshift=-10,rotate=-45] {\texttt{$\pctr_{i,2m-2}$ = 0?}} (r5);
			\draw [->] (r5) edge node [yshift=16,xshift=-10,rotate=-45] {\texttt{$\pctr_{i,2m-1}$ = 0?}} (r6);
			
			\draw [->] (r1) edge node [left,xshift=-3,yshift=-15,rotate=-45,] {\texttt{$\pctr_{i,1}$} $\neq$ 0?} (s1);
			\draw [->] (s1) edge node [left,xshift=-3,yshift=-15,rotate=-45,] {\texttt{$\pctr_{i,1}$ -= 1}} (s2);
			\draw [->] (s2) edge node [left,xshift=-3,yshift=-15,rotate=-45,] {\texttt{$\pctr_{i,2}$ += 1}} (s3);
			\draw [->] (s3) edge node [left,xshift=-3,yshift=-15,rotate=-45,] {$\ctr_{i,1}$ -= 1} (q);

			\draw [->] (r2) edge node [left,xshift=-3,yshift=-15,rotate=-45,] {\texttt{$\pctr_{i,2}$} $\neq$ 0?} (s4);
			\draw [->] (s4) edge node {$\varepsilon$} (s3);
			
			\draw [->] (r5) edge node [left,xshift=-3,yshift=-15,rotate=-45,] {\texttt{$\pctr_{i,2m-1}$} $\neq$ 0?} (r51);
			\draw [->] (r51) edge node [left,xshift=-3,yshift=-15,rotate=-45,] {\texttt{$\pctr_{i,2m-1}$ -= 1}} (r52);
			\draw [->] (r52) edge node [left,xshift=-3,yshift=-15,rotate=-45,] {\texttt{$\pctr_{i,2m}$ += 1}} (dm);
			
			\draw [->] (r6) edge node {$\varepsilon$} (dm);
			
			\draw [->] (dm) |- node [near start,left,xshift=-3,yshift=-15,rotate=-45,] {$\ctr_{i,m}$ -= 1} (q);
			\draw [->] (dm) edge node [above] {$\varepsilon$} (dm1);
			
			\draw [->] (d2) |- node [near start,left,xshift=-3,yshift=-15,rotate=-45,] {$\ctr_{i,2}$ -= 1} (q);
			\draw [->] (d2) edge node [above] {$\varepsilon$} (s3);
			
			\draw [->] (dm1) |- node [near start,left,xshift=-3,yshift=-15,rotate=-45,] {$\ctr_{i,m-1}$ -= 1} (q);
			\draw [->] (dm1) edge node [above] {$\varepsilon$} (t2);
			
			\draw [->] (t2) edge node [above] {$\varepsilon$} (d2);

		\end{tikzpicture}
		\caption{Transition structure to replace a decrementing transition $(p, a, \mathsf{dec}_i, q)$.}
		\label{fig:dec-gadget}
	\end{center}
\end{figure}
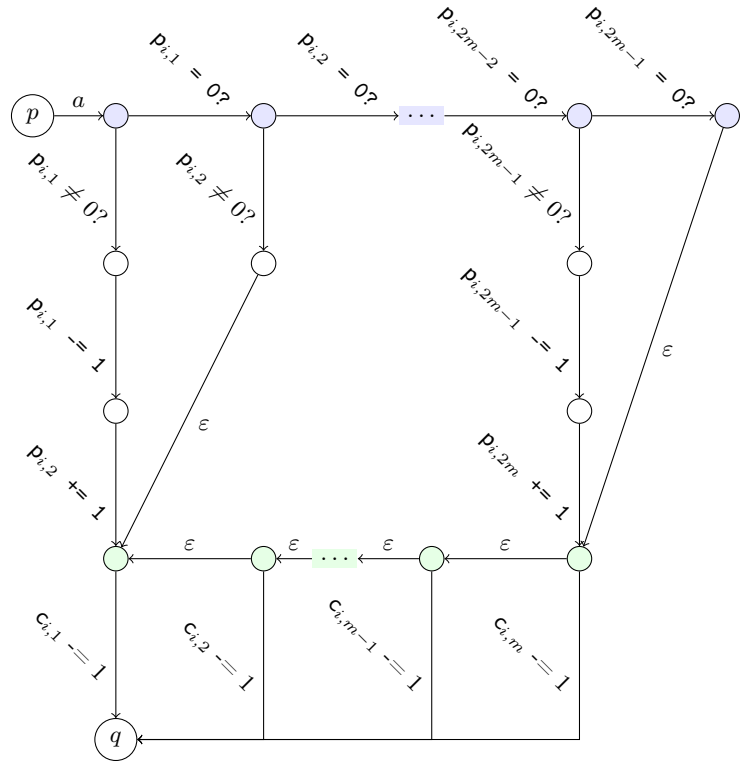

Finally, we consider the zero- (\cref{fig:zero-gadget}) and non-zero-tests (\cref{fig:nonzero-gadget}). We test a counter $\ctr_i$ for zero by making sure \emph{all} partial counters $\ctr_{i,j}$ are zero: we test a counter $\ctr_{i}$ for non-zero by testing if \emph{any} partial counter $\ctr_{i,j}$ is non-zero.

\begin{figure}
	\begin{center}
		\begin{tikzpicture}[state/.style={draw,shape=circle,minimum width=1em},auto, node distance=5em]
			\node [state] (p) {$p$};
			\node [state,right=of p] (p1) {};	
			\node [state,right=of p1] (p2) {};
			\node [right=of p2] (t) {$\ldots$};
			\node [state,right=of t] (p3) {};
			\node [state,right=of p3] (q) {$q$};
			
			\draw [->] (p) edge node {$a$} (p1) {};
			\draw [->] (p1) edge node [yshift=12,xshift=-5,rotate=-45] {$\ctr_{i,1}$ = 0?} (p2);
			\draw [->] (p2) edge node [yshift=12,xshift=-5,rotate=-45] {$\ctr_{i,2}$ = 0?} (t);
			\draw [->] (t) edge node [yshift=16,xshift=-10,rotate=-45] {$\ctr_{i,m-1}$ = 0?} (p3);
			\draw [->] (p3) edge node [yshift=12,xshift=-5,rotate=-45] {$\ctr_{i,m}$ = 0?} (q);
		\end{tikzpicture}
		\caption{Transition structure to replace a zero-test $(p, a, \mathsf{zero}_i, q)$.}
		\label{fig:zero-gadget}
	\end{center}
\end{figure}
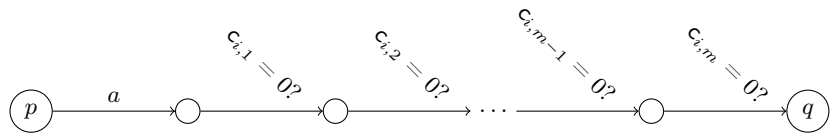

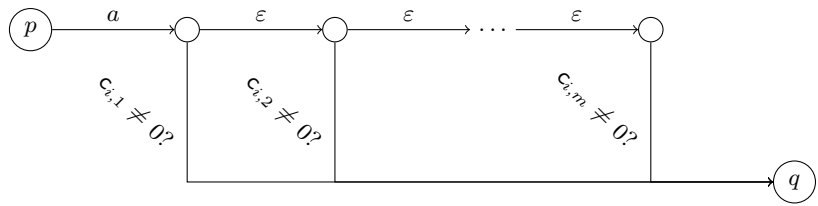
\begin{figure}
	\begin{center}
		\begin{tikzpicture}[state/.style={draw,shape=circle,minimum width=1em},auto, node distance=5em]
			\node [state] (p) {$p$};
			\node [state,right=of p] (p1) {};	
			\node [state,right=of p1] (p2) {};
			\node [right=of p2] (t) {$\ldots$};
			\node [state,right=of t] (p3) {};
			\node [right=of p3] (t') {};
			\node [state,below=of t'] (q) {$q$};
			
			\draw [->] (p) edge node {$a$} (p1) {};
			\draw [->] (p1) |- node [left,near start, yshift=-15,xshift=-5,rotate=-45] {$\ctr_{i,1}$ $\neq$ 0?} (q);
			\draw [->] (p2) |- node [left,near start, yshift=-15,xshift=-5,rotate=-45] {$\ctr_{i,2}$ $\neq$ 0?} (q);
			\draw [->] (t) edge node [] {$\varepsilon$} (p3);
			\draw [->] (p3) |- node [left,near start,yshift=-15,xshift=-5,rotate=-45] {$\ctr_{i,m}$ $\neq$ 0?} (q);
			\draw [->] (p1) edge node {$\varepsilon$} (p2);
			\draw [->] (p2) edge node {$\varepsilon$} (t);
		\end{tikzpicture}
		\caption{Transition structure to replace a non-zero-test $(p, a, \mathsf{nz}_i, q)$.}
		\label{fig:nonzero-gadget}
	\end{center}
\end{figure}

%% file: appendix-unboundedness-reduction.tex
Let us first observe that one can, given a subset $Q$, compute in logspace a context-free grammar for $L_{A,Q}$.
To compute this language, we construct two grammars (one for the $u$ part of $A \derivs[Q] uAv$, 
one for the $v$ part). Then $L_{A,Q}$ is the union of those two grammars. We present 
the construction of the $u$-grammar here, the other can be constructed symmetrically.

Intuitively, for each production in $Q$, we will attempt to guess which of its 
resulting non-terminals will produce the subtree with the dangling $A$. All 
non-terminals to the left will be evaluated normally, everything to the right 
(including terminals) will be discarded, and the selected non-terminal will continue 
guessing which of \emph{its} children will be responsible for producing $A$.

Formally, we construct a grammar as follows: 
\begin{enumerate}
	\item The set of non-terminals is duplicated, $N \coloneqq M \cup \{\hat{X} \mid X \in M\}$
	\item All the original productions are kept, and additionally, we introduce the following productions: If $X \to w$ is a production from $Q$ with at least one non-terminal in $w$, then for all ways $uYv$ to split $w$ along a non-terminal, we introduce the production $\hat{X} \to u\hat{Y}$.
	\item Introduce a production $\hat{A} \to \varepsilon$.
	\item The start symbol is $\hat{A}$.
\end{enumerate}

Observe that any marked non-terminal $\hat{X}$ can only ever occur as the right-most symbol in any (partial) derivation. It can only be removed by the production $\hat{A} \to \varepsilon$. Therefore, a word $u$ is produced by this new grammar if and only if there is a derivation $A \derivs[Q] uAv$ in $Q$.

We also remark the following fact about the structure of the new grammar:

\begin{remark}
	\label{cor:linear-regular-bound}
	If all productions of the original grammar are left-linear (or all are right-linear) then we also get left-linear (respectively right-linear) grammars for the languages $L_{A,Q}$. We can therefore construct NFAs for $L_{A,Q}$, which yields better complexity results for $\Z$-VASS.
\end{remark}

\realizableCancelable*
\begin{proof}
	We begin by showing that we can decide realizability of a set
	$M\subseteq N$ in $\NP$.  In \cite{DBLP:conf/cade/VermaSS05}, it was shown that
	one can compute in polynomial time an existential Presburger formula $\psi$
	such that a vector $\vec{u} \in \N^P$ satisfies $\psi$ if and only if $\vec{u}$
	corresponds to a terminal derivation of a context-free grammar. We augment this
	formula with constraints that state that for all $A \in M$ some production $A
	\to w$ occurs, and no productions $A \to w$ with $A \notin M$ occur; as well as
	a constraint that the total counter effect should be zero, then check for
	satisfiability.
	
	In order to check $M$-realizability of a production $p$, we first
	construct a second grammar $G'$. The grammar $G'$ simulates pumps of $G$. More
	precisely, $G'$ simulates some pump $A_1\derivs u_1A_1v_1$, then some pump
	$A_2\derivs u_2A_2v_2$, etc, and stops after simulating finitely many pumps.
	However, instead of producing the terminals that these pumps generate, it uses
	$P$ as its set of terminals and produces the letter $p'\in P$ whenever it
	applies the production $p'$. Then, since $P$ is the terminal alphabet of $G'$, we have for every $\vec{u}\in\N^P$:
	\[ \vec{u}\in \Parikh(L(G')) ~~\iff~~ \exists k\in\N,~\text{pumps}~\tau_1,\ldots,\tau_k\colon \vec{u}=\Parikh(\tau_1)+\cdots+\Parikh(\tau_k) \]
	Such a grammar $G'$ can clearly be computed in polynomial time. Given
	the grammar $G'$, we apply again the result of \cite{DBLP:conf/cade/VermaSS05}
	to compute an existential Presburger formula $\psi$ that defines
	$\Parikh(L(G'))$. With this, we observe that $p$ is $M$-realizable if and only if
	there exists some $\vec{u}\in \Parikh(L(G'))$ such that
	\begin{enumerate}
	\item $\vec{u}(p)>0$,
	\item $\varphi(\vec{u})=\bzero$, and
	\item all productions occurring in $\vec{u}$ belong to $M$.
	\end{enumerate}
	Since with $\psi$, the existence of such a $\vec{u}$ can easily
	expressed in another existential Presburger formula $\psi'$, it remains to
	check satisfiability of $\psi'$, which can be done in $\NP$.
\end{proof}

%% file: appendix-growth.tex
\subsubsection{The relation $\linembed$ and growth}
Let us show here the following fact used in the main text:
\begin{lemma}
If $L_1\linembed L_2$ and $L_1$ has exponential growth, then $L_2$ has exponential growth as well.
\end{lemma}
\begin{proof}
We first claim that there exists $c>0$
with $|L_1\cap\Sigma^{\le m}|\le cm(m+1)\cdot |L_2\cap \Sigma^{\le cm}|$.  Since
$L_1\linembed L_2$, there exists a $c>0$ such that for every word $w_1\in L_1$,
there exists a word $w_2\in L_2$ such that $|w_2|\le c\cdot |w_1|$ and $w_1$ is
a factor of $w_2$. This allows us to construct an injection
\[ f_m\colon L_1\cap\Sigma^{\le m} \to (L_2\cap\Sigma^{\le cm})\times\{1,\ldots,cm\}\times\{0,\ldots,m\}, \]
which sends each word $w_1\in L_1\cap\Sigma^{\le m}$ to $(w_2,r,|w_1|)$, where
$w_2\in L_2\cap\Sigma^{\le cm}$ is the word existing due to $L_1\linembed L_2$ and $r\in\{1,\ldots,cm\}$
is a position in $w_2$ at which $w_1$ appears. Then $f_m$ must clearly be
injective. Since the set on the right-hand side has cardinality
$cm(m+1)\cdot|L_2\cap \Sigma^{\le cm}|$, this proves the claim.


We can now deduce that $L_2$ has exponential growth: Since $L_1$ has
exponential growth, there is some real $r>1$ such that $|L_1\cap \Sigma^{\le
m}|\ge r^m$ for infinitely many $m$.  By our claim, this implies $|L_2\cap
\Sigma^{\le cm}|\ge \tfrac{r^m}{cm(m+1)}$ for infinitely many $m$.  Since
$\tfrac{r^m}{cm(m+1)}\ge (\sqrt[2c]{r})^{cm}$ for almost all $m$, we have in
particular $|L_2\cap\Sigma^{\le cm}|\ge (\sqrt[2c]{r})^{cm}$ for infinitely
many $m$.  Thus, $L_2$ has exponential growth.
\end{proof}

\subsubsection{Proof of \cref{unboundedness-factors-stronger}}
\begin{proof}
We proceed as in the proof of \cref{unboundedness-factors}.  However, we need
to show that each word $u$,$v$ can be found as a factor of some word $w\in
L(G)$ that is at most linear in $|u|$ or $|v|$, respectively. We show this for
$u$, the case of $v$ is analogous.  First, choose $\ell\in\N$ so that for every
non-terminal $B\in M$, there is a word $x\in\Sigma^*$ with $B\derivs[Q] x$ and
$|x|\le \ell$. Then since $A\derivs[Q] uAv$, we can also find a derivation
$A\derivs[Q] uAv'$, where $|v'|\le \ell\cdot|N|\cdot(|u|+1)$. This is because
in the derivation $A\derivs[Q] uAv$, the path from $A$ to the $A$-leaf has at
most $|u|$ nodes that are \emph{left-branching} (meaning: branch to the left).
By cutting out pumps, we may assume that between any two left-branching nodes
(and above the first, and below the last), there are at most $|N|$ nodes. Thus,
we arrive at a derivation tree for $A\derivs[Q] uAv''$ with at most
$|N|\cdot(|u|+1)$ right-branching nodes.  Then, replacing each subtree under a
right-branching node by a tree deriving a word of length $\le\ell$ yields the
derivation $A\derivs[Q] uAv'$.

Now we apply the proof of \cref{unboundedness-factors} to the derivation $A\derivs[Q]uAv'$.
Observe that the total number of productions used in $\tau_0$, $\tau'_1, \dots, \tau'_m$ and $\tau$ is bounded by $O(|u|+|v'|)=O(|u|)$:
First, note that $\tau_0$ and the flows $\vec{f}_p$ are independent of $u$ and $v'$.
The total number of productions in $\bmf_\tau=\Parikh(\tau'_1)+\cdots+\Parikh(\tau'_m)$
is bounded linearly in the number of productions of $\tau$ by~\eqref{eq:ftau}.
The latter in turn is bounded by $|u|+|v'|$ if $G$ is in Chomsky normal form.
\end{proof}

\subsubsection{Notions of polynomial and exponential growth}
In \cite{DBLP:journals/ijfcs/GawrychowskiKRS10}, polynomial growth for a
language $L\subseteq\Sigma^*$ is defined as the existence of a polynomial
$p(x)$ such that $|L\cap \Sigma^m|\le p(m)$ for every $m\ge 0$. Moreover,
exponential growth is defined as the existence of a real $r>1$ such that
$|L\cap \Sigma^m|\ge r^m$ for infinitely many $m\ge 0$. This differs slightly
from our definition, which is also sometimes called \emph{cumulative growth} \cite{ciobanu2014conjugacy}.
One can show that in fact both definitions are equivalent with respect to polynomial and exponential growth,
see e.g.~\cite[Lemma~2.3]{DBLP:phd/dnb/Ganardi19}.
In the following we provide a proof for completeness.

First, if there is a polynomial $p(x)$ with $|L\cap\Sigma^m|\le p(m)$ for all
$m$, then we also have $|L\cap \Sigma^{\le m}|\le \sum_{i=0}^m
|L\cap\Sigma^i|\le \sum_{i=0}^m p(i)\le (m+1)\cdot p(m)$, which is also a
polynomial bound. (Note that we may clearly assume that $x\mapsto p(x)$ is a
monotone function). Thus, $L$ has polynomial growth by our definition. The
converse is trivial.

Second, suppose there is a real $r>1$ such that $|L\cap \Sigma^{\le m}|\ge r^m$ for
infinitely many $m$. We claim that then $|L\cap\Sigma^{m}|\ge \sqrt{r}^m$
for infinitely many $m$. Towards a contradiction, suppose that there is a $k\ge
0$ such that for all $m\ge k$, we have $|L\cap \Sigma^{m}|<\sqrt{r}^m$. This implies
\begin{equation}
\begin{aligned}
	r^m&\le |L\cap \Sigma^{\le m}|=\sum_{i=0}^m |L\cap \Sigma^i|<M+\sum_{i=0}^m \sqrt{r}^i=M+\frac{\sqrt{r}^{m+1}-1}{\sqrt{r}-1}\label{growth-inequality}
\end{aligned}
\end{equation}
for all $m\ge k$, where $M$ is the constant $\sum_{i=0}^k
|L\cap\Sigma^i|-\sqrt{r}^i$.
The last equality follows from the induction
\[
  \sqrt{r}^0 = 1 = \frac{\sqrt{r} - 1}{\sqrt{r} -1}
  \quad \text{and} \quad
  \sum_{i=0}^{m+1} \sqrt{r}^i \mathrel{\overset{\text{IH}}{=}}
  \frac{\sqrt{r}^{m+1}-1}{\sqrt{r}-1} + \sqrt{r}^{m+1} =
  \frac{\sqrt{r}^{m+2}-1}{\sqrt{r}-1}.
\]
However, \cref{growth-inequality} is clearly
violated for large enough $m$. Again, the converse is obvious: If
$|L\cap\Sigma^m|\ge r^m$, then clearly also $|L\cap\Sigma^{\le m}|\ge r^m$.